\documentclass[12pt, onecolumn]{IEEEtran}
\usepackage{epsfig,latexsym,amssymb,amsmath,graphics,color}
\usepackage{graphicx,subfigure}
\usepackage{float}
\usepackage{verbatim}
\usepackage{multicol}
\usepackage{amstext}
\usepackage{mathrsfs}
\usepackage[all]{xy}
\usepackage{cite}
\usepackage{setspace}
\usepackage{amsthm}

\usepackage{algorithm}
\usepackage{algorithmicx}
\usepackage{algpseudocode}
\usepackage{booktabs}
\usepackage[sans]{dsfont}
\usepackage{txfonts}
\usepackage{upgreek}
\usepackage{multirow,multicol,array,booktabs}

\newtheorem{prop}{Proposition}

\newtheorem{cor}{Corollary}

\newtheorem{lm}{Lemma}

\newtheorem{thm}{Theorem}

\newcommand{\be}{\begin{eqnarray}}
\newcommand{\ee}{\end{eqnarray}}
\newcommand{\benn}{\begin{eqnarray*}}
\newcommand{\eenn}{\end{eqnarray*}}
\def\IR{\rm I \kern-0.20em R}
\newcommand{\utwi}[1]{\mbox{\boldmath $ #1$}}

\newcommand{\bthm}{\begin{thm}}
\newcommand{\ethm}{\end{thm}}

\newcommand{\bcor}{\begin{cor}}
\newcommand{\ecor}{\end{cor}}
\newcommand{\bprop}{\begin{prop}}
\newcommand{\eprop}{\end{prop}}
\newcommand{\blm}{\begin{lm}}
\newcommand{\elm}{\end{lm}}
\newcommand{\beq}{\begin{equation}}
\newcommand{\eeq}{\end{equation}}
\newcommand{\ber}{\begin{eqnarray}}
\newcommand{\eer}{\end{eqnarray}}

\newcommand{\bproof}{\begin{proof}}
\newcommand{\eproof}{\end{proof}}

\newcommand{\bit}{\begin{itemize}}
\newcommand{\eit}{\end{itemize}}
\newcommand{\ben}{\begin{enumerate}}
\newcommand{\een}{\end{enumerate}}
\newcommand{\bdesc}{\begin{description}}
\newcommand{\edesc}{\end{description}}
\newcommand{\beqarrn}{\begin{eqnarray*}}
\newcommand{\eeqarrn}{\end{eqnarray*}}
\newcommand{\bproofof}{\begin{proofof}}
\newcommand{\eproofof}{\end{proofof}}
\newenvironment{rem}{\begin{trivlist}\item[]{\bf
Remark:}\hspace{4mm}}{\end{trivlist}}
\newcommand{\brem}{\begin{rem}}
\newcommand{\erem}{\end{rem}}
\newenvironment{rems}{\begin{trivlist}\item[]{\bf
Remarks}\begin{itemize}}{\end{itemize}\end{trivlist}}
\newcommand{\brems}{\begin{rems}}
\newcommand{\erems}{\end{rems}}
\newtheorem{fact}{Fact}
\newcommand{\bfact}{\begin{fact}}
\newcommand{\efact}{\end{fact}}
\newtheorem{examp}{Example}
\newcommand{\bexamp}{\begin{examp}\rm}
\newcommand{\eexamp}{\end{examp}}
\newtheorem{defn}{Definition}
\newcommand{\bdefn}{\begin{defn}\rm}
\newcommand{\edefn}{\end{defn}}

\newtheorem{alg}{Algorithm}
\newcommand{\balg}{\begin{alg}}
\newcommand{\ealg}{\end{alg}}

\newtheorem{prob}{Problem}
\newcommand{\bprob}{\begin{prob}}
\newcommand{\eprob}{\end{prob}}

\newcommand{\bvtm}{\begin{verbatim}}
\newcommand{\bfig}{\begin{figure}}
\newcommand{\efig}{\end{figure}}
\newcommand{\bcen}{\begin{center}}
\newcommand{\ecen}{\end{center}}

\long\def\comment#1{}

\def \n2{{N_0 \over 2}}

\def \h5{\hspace{0.5in}}

\newcommand{\bz}{{\utwi{z}}}

\newcommand{\bN}{{\utwi{N}}}

\newcommand{\bS}{{\utwi{S}}}

\newcommand{\bZ}{{\utwi{Z}}}

\title{Multi-layer Superimposed Transmission with Symbol Boundary Offset for Optical Wireless Scattering Communication}

\author{Guanchu Wang, Chen Gong, Zhimeng Jiang and Zhengyuan Xu
\thanks{This work was supported by Key Program of National Natural Science
Foundation of China (Grant No. 61631018) and Key Research Program of Frontier
Sciences of CAS (Grant No. QYZDY-SSW-JSC003).}
\thanks{The authors are with Key Laboratory of Wireless-Optical Communications, Chinese Academy of Sciences, School of Information Science and Technology,
University of Science and Technology of China, Hefei, China.
Email: \{hegsns, zhimengj\}@mail.ustc.edu.cn, \{cgong821, xuzy\}@ustc.edu.cn.}
}

\begin{document}

\newtheorem{Proposition}{Proposition}
\newtheorem{Lemma}{Lemma}
\newtheorem{Theorem}{Theorem}

\begin{spacing}{2.0}

\maketitle

\begin{abstract}
We investigate the multi-layer superimposed transmission for optical wireless scattering communication where the symbol boundaries on different signal layers are not necessarily aligned in the time domain.
We characterize the multi-layer transmission based on a hidden markov model.
Then, we obtain the achievable rates of all signal layers and a single layer, and provide a numerical solution.
Furthermore, we propose approaches on the channel estimation as well as joint symbol detection and decoding.
Finally, both simulations and experiments are conducted to evaluate the performance of the proposed approaches, and validate the feasibility of the proposed transmission and signal detection approaches.
\end{abstract}


\section{Introduction}


Non-line of sight (NLOS) Ultra-violet (UV) scattering communication serves as a good candidate for the applications where radio-silence is required and the transmitter-receiver alignment is hard to guarantee due to obstacles or the user mobility.
\textcolor{red}{
Moreover, it is promising for outdoor communication under strong solar background because of negligible solar radiation in the UV spectrum \cite{xu2008ultraviolet}.}
\textcolor{red}{
Theoretical analysis \cite{xu2015effects}, numerical simulation \cite{sun2016closed} and real experiments \cite{raptis2016power, wang2017turbulence} show an extremely large path loss between the transmitter and receiver, where the received signal can be detected by photon-counting receiver and characterized by Poisson distributed number of discrete photoelectrons.}

\textcolor{red}{
The capacity of point-to-point continuous-time Poisson channel has been investigated in \cite{davis1980capacity, wyner1988capacity, frey1991information} and the capacity of discrete-time Poisson channel has been derived in \cite{lapidoth2009capacity, lapidoth2011discrete}.}
Based on the Poisson channel model, several types of channel model such as Poisson fading \cite{chakraborty2007poisson}, MIMO \cite{chakraborty2008outage}, interfering \cite{lai2015capacity},  broadcast \cite{kim2016superposition} , and multiple access \cite{mehravari1984poisson, lapidoth1998poisson, bross2001error} channels have been studied in recent years.
\textcolor{red}{
Specifically, code-division and non-orthogonal multiple transmission has been studied in
\cite{wang2018signal}, and random access packet-switched systems was proposed in \cite{raychaudhuri1981performance}.
}
Other existing works on NLOS UV scattering communication based on the Poisson and extended channel model are the channel link gain with impulse response \cite{ding2011correction, zuo2013closed}, channel estimation with inter-symbol interference \cite{gong2015channel}, signal detection with receiver diversity \cite{gong2015lmmse}, and the relay protocol \cite{gong2015non}.

\textcolor{red}{
In this work, we characterize multi-layer superimposed transmission in discrete Poisson channel, where the transmitted symbols in various layers are superimposed, and the symbol boundaries on different signal layers are not necessarily aligned.}
Specifically, we adopt hidden markov model (HMM) \cite{rabiner1986introduction, ephraim2002hidden} to characterize the superimposed channel.
Then, we conceive the achievable transmission rates for all signal layers and a single layer, and obtain the exact and approximated solution\cite{liu2018hidden}.
For receiver-side signal processing, we propose channel estimation based on expectation-maximization (EM) algorithm \cite{moon1996expectation, dempster1977maximum}, and adopt Viterbi \cite{forney1973viterbi} and Bahl-Cocke-Jelinek-Raviv (BCJR) \cite{bahl1974optimal} algorithms for symbol detection.
Furthermore, we propose iterative algorithm for maximum-likelihood/maximum a posteriori probability (ML/MAP) joint decoding \cite{hagenauer1996iterative, benedetto1997soft}.
Finally, we conduct offline experiments to evaluate the performance of the proposed approaches.
It is seen that based on the experimental measurements, the proposed approaches perform close to the simulation results with identical channel parameters.


The remainder of the paper is organized as follows.
In Section II, we characterize the superimposed NLOS scattering communication using HMM.
In Section III, we investigate the achievable transmission rates and obtain a numerical solution on the achievable transmission rate of all signal layers and a single layer.
In Section IV, we propose the EM-based channel estimation as well as joint symbol detection and decoding.
Numerical and experimental results are given in Sections V and VI, respectively.
Finally, we conclude this paper in Section VII.

\section{System Model}

\subsection{Superimposed Transmission based on Discrete Poisson Asynchronous Channel}


We consider a NLOS scattering communication system adopting on-off key (OOK) modulation that outperforms pulse-position modulation (PPM, please refer to Appendix A for more details).
The overall transmission signal can be split into multiple signal layers which are superimposed possibly in an asynchronous manner, i. e., the symbols in different layers are not necessarily aligned.
As shown in Figure \ref{fig:Channel_parallelism}, the overall transmission can be split into $L$ signal layers, denoted as layer $1, 2, ..., L$.
Let $M$ denote the number of transmitted symbols in each single layer; $T_s$ denote the symbol duration; and $\rho_1, \rho_2, \dots, \rho_L$ denote the normalized relative delay in terms of $T_s$ $\big( \sum_{i=1}^{L}\rho_i = 1 \big)$, where $\rho_i$ denotes the normalized delay between layer $i$ and layer $i+1$, for $1 \leq i \leq L$ (here layer $L+1$ equals Layer $1$).

\begin{figure}
\centering
	\includegraphics[width=0.6\textwidth]{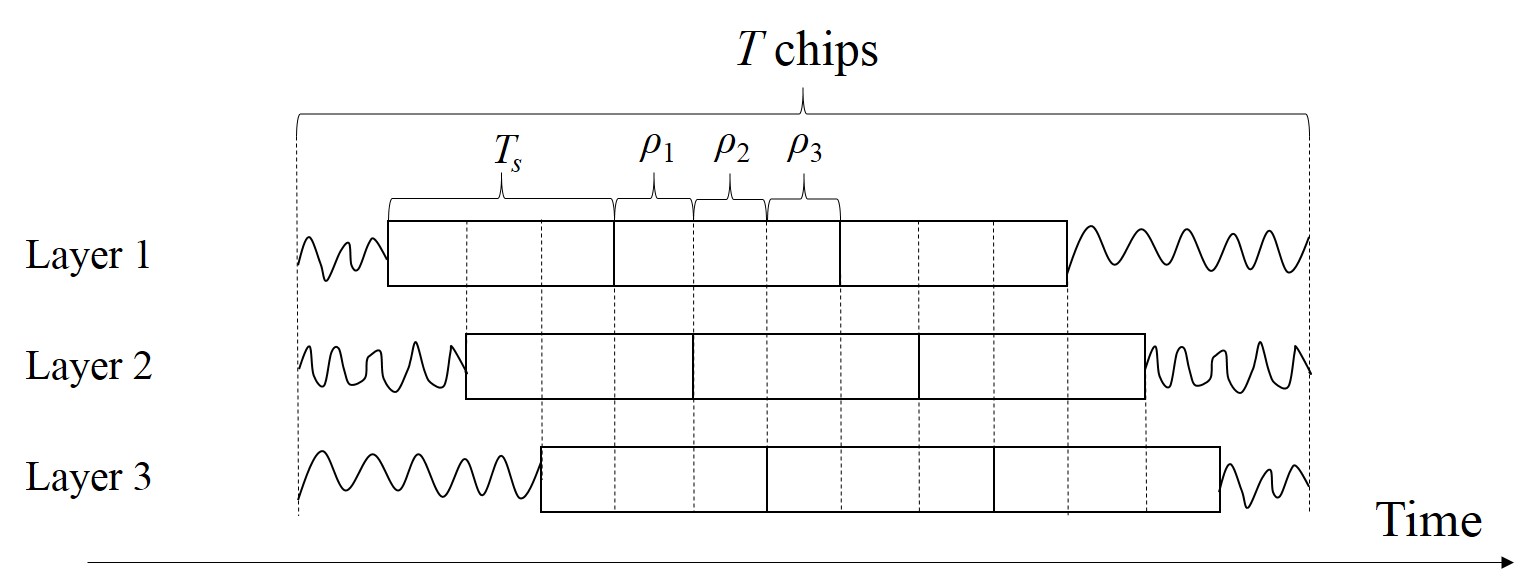}
    \caption{\label{fig:Channel_parallelism} Illustration for $3$-layer superimposed transmission.}
\end{figure}

In order to characterize the symbol duration offset in different signal layers, we divide the symbols in different signal layers into chips subjected to symbol boundaries, where the symbol detection is performed based on the received signal in each chip.
The symbol misalignment and relative delay are illustrated in Figure \ref{fig:Channel_parallelism}, where $T$ denotes the number of overall chips and $T = ML + L - 1$.

Due to the weak received signal intensity of NLOS scattering communications, the received signal can be characterized by discrete photoelectrons, whose number satisfies a Poisson distribution.
More specifically, let $\lambda_1, \lambda_2, \dots, \lambda_L$ denote the mean number of detected photoelectrons in each symbol duration, and
\textcolor{red}{$\boldsymbol{z}_1, \boldsymbol{z}_2, \dots, \boldsymbol{z}_L$ denote the transmitted binary symbols in the $L$ signal layers, where $\bz_{i} = [z_{i, 1}, z_{i, 2}, ..., z_{i, M}] \in \{ 0, 1\}^M$; $z_{i, m}$ demotes the $m^{\text{th}}$ symbol in layer $i$ for $1 \leq i \leq L$ and $1 \leq m \leq M$; and the transmitted symbols are independent of each others.}
The number of detected photoelectrons $N_t$ in the $t$-th chip for $1 \leq t \leq T$ satisfies the following Poisson distribution
\begin{equation}
\begin{aligned}
\label{eq:Poisson_channel}
\mathbb{P} (N_t = n) = \frac{\tau^n_t}{n !}( \lambda_0 + \boldsymbol{\Lambda}^T \boldsymbol{S}_t )^n e^{-\tau_t (\lambda_0 + \boldsymbol{\Lambda}^T \boldsymbol{S}_t)},
\end{aligned}
\end{equation}
where $\boldsymbol{\Lambda} = [ \lambda_1, \lambda_2, \dots, \lambda_M ]^T$; $\tau_t = \rho_{(t-1   \text{ mod } M) + 1}$; $\boldsymbol{S}_t = [z_{1, \lceil \frac{t}{L} \rceil}, z_{2, \lceil \frac{t-1}{L} \rceil}, \dots, z_{L, \lceil \frac{t-L+1}{L} \rceil}]^T$; $z_{i, 0} = 0, z_{i, M+1} = 0$ for $1 \leq i \leq L$; and $\lambda_0$ denotes the mean number of background radiation photoelectrons in a symbol duration.

\subsection{Hidden Markov Model for Asynchronous Signal Superposition}

Due to the overlap of different layers, the numbers of detected photoelectrons in adjacent chips are correlated with each other.
In the $t$-th chip, $N_t$ depends on $\boldsymbol{S}_t$, which depends on $\boldsymbol{S}_{t-1}$.
Consequently, we can adopt HMM to characterize the signal model in the chip level.

We denote $\varmathbb{S}_T = \{ \boldsymbol{S}_t | 1 \leq t \leq T \}$ and $\boldsymbol{N}_T = [ N_1, N_2, \cdots, N_T] \in \mathbb{N}^L$ as the state and observation sequences of the HMM, respectively, where $\boldsymbol{S}_t \in \mathscr{B}^{L}$, and $\mathscr{B}^{L}$ denotes the state space of the $t$-th chip given by
\begin{equation}
\begin{aligned}
\mathscr{B}^{L} = \bigg\{ \sum_{i=1}^L \theta_i \boldsymbol{e}_i \mid \theta_i \in \{ 0, 1\}, 1 \leq i \leq L \bigg\},
\end{aligned}
\end{equation}
where $\boldsymbol{e}_i$ denotes the $i$-th column of $L\times L$ identity matrix.

An HMM is determined by parameters $(\boldsymbol{\pi}_1, \boldsymbol{A}_t, \boldsymbol{B}_t)$, where $\boldsymbol{\pi}_1$, $\boldsymbol{A}_t$ and $\boldsymbol{B}_t$ denote the initial distribution, state transition matrix and observation emission matrix, respectively.
Note that the initial state depends on the first symbol in the first layer, thus $\boldsymbol{\pi}_1$ is given by
\begin{equation}
\begin{aligned}
\label{eq:HMM_pi1}
\boldsymbol{\pi}_1 &=
\Big\{ q_{1,1}, 1-q_{1,1},0,0,\dots,0 \Big\},
\end{aligned}
\end{equation}
where $q_{i,j} = \mathbb{P}(z_{i,j} = 1)$ denotes the prior possibility of symbol $z_{i,j}$ for $1 \leq i \leq L$ and $1 \leq j \leq M$.
The symbols in the same signal layer may have different prior probabilities since they may be allocated to different users.

The state transition matrix is given by $\boldsymbol{A}_t = \Big[ a_{t,i,j} | \boldsymbol{s}_{t,i} \in \mathscr{B}^{L}, \boldsymbol{s}_{t+1,j} \in \mathscr{B}^{L} \Big]$, where each element $a_{t,i,j}$ is given by
\begin{equation}
\begin{aligned}
\label{eq:atij}
a_{t,i,j} &= \mathbb{P} (\boldsymbol{S}_{t+1} = \boldsymbol{s}_{t+1,j} | \boldsymbol{S}_t = \boldsymbol{s}_{t,i})
= q_{k, \lceil \frac{t-k+2}{L} \rceil}^{ \boldsymbol{s}_{t+1,j} \cdot \boldsymbol{e}_{k}} (1-q_{k, \lceil \frac{t-k+2}{L} \rceil})^{\boldsymbol{s}_{t+1,j} \cdot \boldsymbol{e}_{k}}
\prod_{r \neq k} \big( \boldsymbol{s}_{t+1,j} \cdot \boldsymbol{e}_r \big) \odot \big( \boldsymbol{s}_{t,i} \cdot \boldsymbol{e}_r \big),
\end{aligned}
\end{equation}
and $k = (t \text{ mod }L) + 1$, which means $\boldsymbol{A}_t$ is cyclical of period $L$; $\boldsymbol{s}_{t,i}, \boldsymbol{s}_{t+1,j} \in \mathscr{B}^{L}$ take values among all possible choices of $\boldsymbol{S}_t$ and $\boldsymbol{S}_{t+1}$, respectively; Moreover, $\odot$ indicates binary logical XNOR.

\begin{proof}
Please refer to Appendix B.
\end{proof}

The observation emission matrix is given by $\boldsymbol{B}_t = \Big[b_{t,i,n+1} | \boldsymbol{s}_{t,i} \in \mathscr{B}^{L}, n \in \mathbb{N} \Big]$, where based on Equation (\ref{eq:Poisson_channel}) each element $b_{t,i,n+1}$ is given by
\begin{equation}
\begin{aligned}
b_{t,i,n+1} = \mathbb{P} (N_t = n | \boldsymbol{S}_t = \boldsymbol{s}_{t,i}) = \frac{\tau_t^n}{n !}(\lambda_0 + \boldsymbol{\Lambda}^T \boldsymbol{s}_{t,i}) e^{-\tau_t (\lambda_0 + \boldsymbol{\Lambda}^T \boldsymbol{s}_{t,i})}.
\end{aligned}
\end{equation}

\subsection{Modeling System With Superimposed Communication}

\textcolor{red}{
The superimposed transmission can be applied to multi-user communication.
Let $K$ denote the number of users.
For $K \leq L$, we can assign each signal layer or multiple layers to one user.
For $K > L$, some users have to share a common signal layer.
Figure \ref{fig:MAC_User} illustrates the scenario with $5$ users sharing $2$ layers via time-division.
\begin{figure}
\centering
	\includegraphics[width=0.8\textwidth]{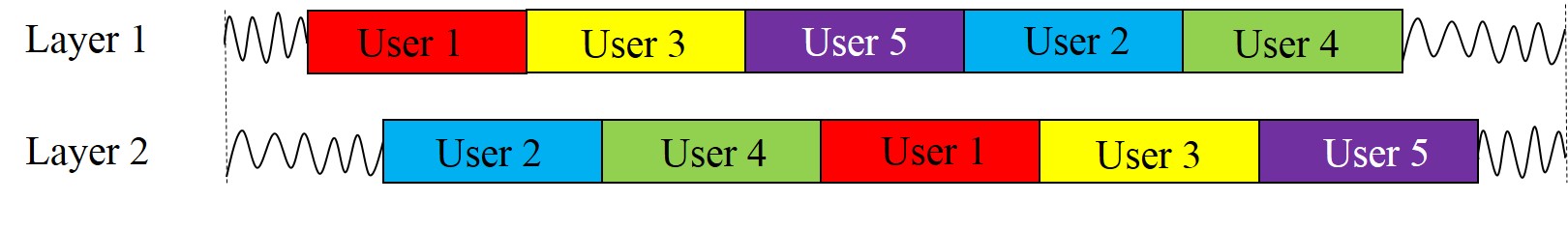}
    \caption{\label{fig:MAC_User} Two-layer transmission with five users.}
\end{figure}
}

\section{Achievable Transmission Rate}

We consider the achievable rates for HMM, and give a numerical solution to the achievable transmission rate of asynchronous signal superposition.

\subsection{Achievable Rates for HMM}

The achievable rates can be derived based on the mutual information between hidden states and observation sequences for HMM.
Let $\mathcal{L} = \{ 1, 2, \dots, L\}$ denote the entire set of signal layers; $\mathcal{U} \subset \mathcal{L}$ denote a subset of layers; and $\varmathbb{Z}_{\mathcal{U}} = \{\boldsymbol{z}_k | k\in\mathcal{U} \}$ denote the set of transmission symbols in layer set $\mathcal{U}$.
Due to the statistical independence of different transmitted symbols, the entropy and conditional entropy of transmitted symbols are given as follows,
\textcolor{red}{
\begin{equation}
\begin{aligned}
\label{eq:entropy_transmitted}
\mathrm{H} (\varmathbb{Z}_{\mathcal{L}})
&= \sum^L_{i=1} \sum^M_{j=1} \mathcal{H} (q_{i,j}),
\\
\mathrm{H} (\varmathbb{Z}_{\mathcal{U}} | \varmathbb{Z}_{\mathcal{L}\setminus \mathcal{U}})
&= \sum_{i\in\mathcal{U}} \sum^M_{j=1} \mathcal{H} (q_{i,j}),
\end{aligned}
\end{equation}
}
where $\mathcal{H} (x) = -x\log_2 x - (1-x)\log_2 (1-x)$.

The entropy and conditional entropy of the transmitted symbols given the observation sequences are given by
\begin{equation}
\begin{aligned}
\label{eq:conditional_entropy}
\mathrm{H} (\varmathbb{Z}_{\mathcal{L}} | \boldsymbol{N}_T) \! &= \!
- \mathbb{E}_{\boldsymbol{z} \in \mathscr{B}^T \atop \boldsymbol{n} \in \mathbb{N}^T} \! \Big[ \! \log_2 \! \mathbb{P} ( \varmathbb{Z}_{\mathcal{L}} \!=\! \boldsymbol{z} | \boldsymbol{N}_T \!=\! \boldsymbol{n} ) \Big],
\\
\mathrm{H} (\varmathbb{Z}_{\mathcal{U}} | \varmathbb{Z}_{\mathcal{L}\setminus \mathcal{U}}, \boldsymbol{N}_T) \! &= \! - \mathbb{E}_{\boldsymbol{z} \in \mathscr{B}^T \atop \boldsymbol{n} \in \mathbb{N}^T} \! \Big[ \! \log_2 \! \mathbb{P} ( \varmathbb{Z}_{\mathcal{U}} \!=\! \boldsymbol{z}_{\mathcal{U}} | \varmathbb{Z}_{\mathcal{L}\setminus \mathcal{U}} \!=\! \boldsymbol{z}_{\mathcal{L}\setminus \mathcal{U}}, \boldsymbol{N}_T \!=\! \boldsymbol{n} ) \Big] ,
\end{aligned}
\end{equation}
where $\mathbb{N}$ denotes the set of natural number; and ${\Omega}^T$ denotes the $T$-time expansion of set ${\Omega}$.

\textcolor{red}{
Note that for $\forall \mathcal{U} \subseteq \mathcal{L}$, $\mathcal{U} \neq \varnothing$, the overall achievable rate \cite{cover2012elements} of the signal layers in set $\mathcal{U}$ must satisfy}
\textcolor{red}{
\begin{equation}
\begin{aligned}
\label{eq:leq_achievable_rates}
\sum_{k \in \mathcal{U}} R_k \leq \frac{1}{M} \mathrm{I} (\varmathbb{Z}_{\mathcal{U}}; \boldsymbol{N}_T | \varmathbb{Z}_{\mathcal{L}\setminus \mathcal{U}}),
\end{aligned}
\end{equation}
}
where coefficient $1/M$ is due to the fact of $M$ symbols in the Markov chain for each signal layer, and $R_k$ denotes the achievable rate of layer $k$; and $\mathrm{I} (\varmathbb{Z}_{\mathcal{U}}; \boldsymbol{N}_T | \varmathbb{Z}_{\mathcal{L}\setminus \mathcal{U}})$ denotes the conditional mutual information given by
\begin{equation}
\begin{aligned}
\label{eq:mutual_information}
\mathrm{I} (\varmathbb{Z}_{\mathcal{U}}; \boldsymbol{N}_T | \varmathbb{Z}_{\mathcal{L}\setminus \mathcal{U}}) = \mathrm{H} (\varmathbb{Z}_{\mathcal{U}} | \varmathbb{Z}_{\mathcal{L}\setminus \mathcal{U}}) - \mathrm{H} (\varmathbb{Z}_{\mathcal{U}} | \varmathbb{Z}_{\mathcal{L}\setminus \mathcal{U}}, \boldsymbol{N}_T).
\end{aligned}
\end{equation}
Letting $\mathcal{U} = \{ k \}$ and $\mathcal{U} = \mathcal{L}$, we have the following two achievable rates of the asynchronous signal superposition,
\begin{equation}
\begin{aligned}
\label{eq:characteristics_parallel_channels}
R^{*}_{k} = \sup R_k = \frac{1}{M} \mathrm{I} (\bZ_k; \bN_T | \bZ_{\mathcal{L}\backslash k}),
\\
R^{*}_{\Sigma} = \sup \sum_{k=1}^L R_k = \frac{1}{M} \mathrm{I} (\varmathbb{Z}_{\mathcal{L}}, \boldsymbol{N}_T),
\end{aligned}
\end{equation}
where $R^{*}_{k}$ and $R^{*}_{\Sigma}$ denote the maximum single-user rate and sum user rate, respectively.

\subsection{Maximum Achievable Transmission Rate for a Single Layer}

We give an algorithm to obtain the maximum achievable rate of a signal layer $R^{*}_{k}$ for $1 \leq k \leq L$.
According to Equation (\ref{eq:characteristics_parallel_channels}), we have
\begin{equation}
\begin{aligned}
R^{*}_{k} = \frac{1}{M} \mathrm{I} (\bZ_k; \bN_T | \bZ_{\mathcal{L}\backslash k}) = \frac{1}{M} \sum_{i=1}^{M} \mathcal{H} (q_{k,i}) - \frac{1}{M} \mathrm{H} (\bZ_k | \bZ_{\mathcal{L}\backslash k}, \bN_T).
\end{aligned}
\end{equation}

We have the following propositions on $\mathrm{H} (\bZ_k | \bZ_{\mathcal{L}\backslash k}, \bN_T)$.

\begin{Proposition}
The chain rule on the conditional probabilities are given as follows
\begin{equation}
\begin{aligned}
\label{eq:Prop1}
\mathbb{P} (\boldsymbol{Z}_k | \boldsymbol{Z}_{\mathcal{L}\backslash k}, \boldsymbol{N}_T) = \prod^{M}_{j=1} \mathbb{P} \big(Z_{k, j} | \{ Z_{i,\lceil \frac{t_j-i+1}{L} \rceil} \}, \{ N_{t_j} \} \big),
\end{aligned}
\end{equation}
where $i \in \mathcal{L} \backslash k$, $k+(j-1)L \leq t_j \leq k+jL-1$;
and $\mathbb{P} \big(Z_{k, j} | \{ Z_{i,\lceil \frac{t_j-i+1}{L} \rceil} \}, \{ N_{t_j} \} \big)$ is the conditional probability of $Z_{k,j}$ given sets $\{Z_{i, \lceil \frac{t_j-i+1}{L} \rceil}\}$ and $\{N_{t_j}\}$.
\end{Proposition}
\begin{proof}
Please refer to Appendix C.
\end{proof}

\begin{Proposition}

The conditional entropy of a single layer is given by
\begin{equation}
\begin{aligned}
\label{eq:2users_conditional_entropy}
\mathrm{H} (\boldsymbol{Z}_k | \boldsymbol{Z}_{\mathcal{L}\backslash k}, \boldsymbol{N}_T)
&= \sum^{M}_{j=1} \sum_{ Z_{k,j} \in \mathscr{B}} \Big( \!\!\! \sum_{Z_{i, \lceil (t-i+1)/L \rceil}\in\mathscr{B}} \!\!\!\!\!\! \Big)^{k+(j-1)L \leq t \leq k+jL-1 \hfill \atop i\in\mathcal{L}\backslash k \hfill} \!\!\!\!\!\!\!\!\!\!\!\!\!\!\!\! \mathbb{P}(Z_{k,j}) \Bigg[ \prod_{ i\in\mathcal{L}\backslash k} \prod_{t=k+(j-1)L}^{k+jL-1} \mathbb{P} (Z_{i,\lceil \frac{t-i+1}{L} \rceil}) \Bigg] \sum_{ \{ N_{t_j} \} \in \mathbb{N}^L} \Bigg[ \prod_{t=k+(j-1)L}^{k+jL-1} \mathbb{P} (N_t | Z_{k,j}, \{Z_{i,\lceil \frac{t-i+1}{L} \rceil}\}) \Bigg]
\\
&\log_2  \frac{\mathbb{P} \big(Z_{k,j}\big) \prod_{t=k+(j-1)L}^{k+jL-1} \mathbb{P} (N_t | Z_{k,j}, \{Z_{i,\lceil \frac{t-i+1}{L} \rceil}\})}{\sum_{ Z_{k,j} \in \mathscr{B}} \mathbb{P} \big(Z_{k,j}\big) \prod_{t=k+(j-1)L}^{k+jL-1} \mathbb{P} (N_t | Z_{k,j}, \{Z_{i,\lceil \frac{t-i+1}{L} \rceil}\})},
\end{aligned}
\end{equation}
where $(\sum_{Z_{i,t}\in\mathscr{B}})^{i \in \{ \phi_1, \phi_2, \ldots\} \atop t \in \{ \omega_1, \omega_2, \ldots\}}$ is the abbreviation of $\sum_{Z_{\phi_1, \omega_1}\in\mathscr{B}} \sum_{Z_{\phi_1, \omega_2}\in\mathscr{B}} \ldots
\sum_{Z_{\phi_2, \omega_1}\in\mathscr{B}} \sum_{Z_{\phi_2, \omega_2}\in\mathscr{B}} \ldots$; $\mathbb{P} (z_{k,j}) = q_{k,j}^{z_{k,j}}(1-q_{k,j})^{(1-z_{k,j})}$; and
\begin{equation}
\begin{aligned}
\label{eq:Pr_N__X_Y}
\mathbb{P} (N_t| Z_{k,j}, \{ Z_{i,\lceil \frac{t-i+1}{L} \rceil} \}) &= \frac{\tau_t^{N_t}}{N_t !} \bigg( \lambda_0 + \lambda_k Z_{k,j} + \sum_{i \in \mathcal{L}\setminus k}\lambda_i Z_{i,\lceil \frac{t-i+1}{L} \rceil} \bigg)^{N_t} e^{-\tau_t (\lambda_0 + \lambda_k Z_{k,j} + \sum_{i \in \mathcal{L}\setminus k}\lambda_i Z_{i,\lceil \frac{t-i+1}{L} \rceil})}.
\end{aligned}
\end{equation}
\end{Proposition}
Specifically, for single user communication the prior probability of transmitted symbols remains constant, i. e, $q_{i,j} = q$ for $1 \leq i \leq L$ and $1 \leq j \leq M$, and the entropy of a single layer can be further simplified into
\begin{equation}
\begin{aligned}
\label{eq:2users_conditional_entropy}
\mathrm{H} (\boldsymbol{Z}_k | \boldsymbol{Z}_{\mathcal{L}\backslash k}, \boldsymbol{N}_T)
&= M \sum_{ Z_{k} \in \mathscr{B}} \bigg( \sum_{ Z_{i,2} \in \mathscr{B}} \sum_{ Z_{i,3} \in \mathscr{B}} \bigg)^{1 \leq i < k}
\bigg( \sum_{ Z_{i,1} \in \mathscr{B}} \sum_{ Z_{i,2} \in \mathscr{B}} \bigg)^{k < i \leq L} \mathbb{P}(Z_{k}) \Bigg[ \prod_{1 \leq i < k} \mathbb{P} (Z_{i,2}) \mathbb{P} (Z_{i,3}) \prod_{k < i \leq L} \mathbb{P} (Z_{i,1}) \mathbb{P} (Z_{i,2}) \Bigg]
\\
&\!\!\!\!\!\!\!\!\sum_{ \{ N_{k+L}, \ldots, N_{k+2L-1} \} \in \mathbb{N}^L} \Bigg[ \prod_{t=k+L}^{k+2L-1} \mathbb{P} (N_t | Z_k, \{Z_{i,\lceil \frac{t-i+1}{L} \rceil}\}) \Bigg]
\log_2  \frac{\mathbb{P} \big(Z_{k}\big) \prod_{t=k+L}^{k+2L-1} \mathbb{P} (N_t | Z_{k}, \{Z_{i,\lceil \frac{t-i+1}{L} \rceil}\})}{\sum_{ Z_{k} \in \mathscr{B}} \mathbb{P} \big(Z_{k}\big) \prod_{t=k+L}^{k+2L-1} \mathbb{P} (N_t | Z_{k}, \{Z_{i,\lceil \frac{t-i+1}{L} \rceil}\})}.
\end{aligned}
\end{equation}
\begin{proof}
Please refer to Appendix D.
\end{proof}

\subsection{Maximum Achievable Sum Rate}

We give an algorithm to obtain the achievable sum rate $R^{*}_{\Sigma}$.
According to Equation (\ref{eq:characteristics_parallel_channels}), we have
\begin{equation}
\begin{aligned}
R^{*}_{\Sigma} = \frac{1}{M} \mathrm{I} (\varmathbb{Z}_{\mathcal{L}}; \boldsymbol{N}_T) = \frac{1}{M} \sum_{i=1}^{M} \sum_{k=1}^{L} \mathcal{H} (q_{k,i}) - \frac{1}{M} \mathrm{H} (\varmathbb{Z}_{\mathcal{L}} | \boldsymbol{N}_T).
\end{aligned}
\end{equation}
Note that the computational complexity of $\mathrm{H} (\varmathbb{Z}_{\mathcal{L}} | \boldsymbol{N}_T)$ grows exponentially with $T$ due to exhaustive enumeration of the state and observation sequences in $\mathscr{B}^{T}$ and $\mathbb{N}^T$.
Consequently, brute-force computation on the exact value is intractable for large $T$.
\textcolor{red}{
The computational complexity can be reduced via sampling $\mathscr{B}^{T}$ and $\mathbb{N}^T$, and the solution for conditional entropies can be approximated by the empirical mean according to the following equation,}
\begin{equation}
\begin{aligned}
\label{eq:conditional_entropy_app}
\mathrm{H} (\varmathbb{Z}_{\mathcal{L}} | \boldsymbol{N}_T) \thickapprox \mathbb{E}_{\boldsymbol{z} \in \Psi_{z} \atop \boldsymbol{n} \in \Psi_n} \Big[ \mathrm{H} (\varmathbb{Z}_{\mathcal{L}} | \boldsymbol{N}_T = \boldsymbol{n}) \Big],
\end{aligned}
\end{equation}
where $\Psi_{z} \subset \mathscr{B}^{T}$ and $\Psi_n \subset \mathbb{N}^T$ denote the set of sufficiently many samples on $\mathscr{B}^{T}$ and $\mathbb{N}^T$ such that the empirical mean becomes converged, respectively.

We resort to Monte Carlo method, which keeps generating random states and observation sequences based on the initial state distribution, the transition probability matrices, and the observation emission matrices.
For each state and observation sequence realization, we have that
\begin{equation}
\begin{aligned}
\label{eq:H_S__N}
\mathrm{H} (\varmathbb{Z}_{\mathcal{L}} | \boldsymbol{N}_T = \boldsymbol{n}) = \mathrm{H} (\varmathbb{S}_{T} | \boldsymbol{N}_T = \boldsymbol{n}),
\end{aligned}
\end{equation}
where efficient computation of the conditional entropy in Equation (\ref{eq:H_S__N}) can be conducted following \cite{hernando2005efficient}.

\subsection{Power Allocation of Overlapped Transmission}
\textcolor{red}{
We regard the achievable rate as the objective function of power allocation.
Generally, the practical issue can be summarized as the following two cases.
}

\textcolor{red}{
\textbf{Case 1:} Given $\lambda_s$, maximize the sum achievable rate $ R^{*}_{\Sigma}$, subject to $\sum_{k=1}^{L} \lambda_k = \lambda_s$.
}

\textcolor{red}{
\textbf{Case 2:} Given $\lambda_s, i$ and $R_{jinf}$, for $1 \leq j \leq L, j \neq i$, maximize the achievable rate $R^*_{i}$ of layer $i$, subject to $R^*_{j} \geq R_{jinf}$ and $\sum^L_{k=1} \lambda_k = \lambda_s$.
}

The numerical solution for $L = 2$ is provided in Section VI.C.

\section{Channel Estimation and Symbol Detection}

We present the receiver-side signal processing including channel estimation, symbol detection as well as joint detection and decoding.

\subsection{Channel Estimation Algorithm}

We can employ pilot sequences to estimate the mean number of detected photoelectrons of each state. However, considering the pilots on all signal layers, the overhead is still non-negligible.
\textcolor{red}{
In this work, the channel estimation can be performed based on pilot sequences on certain signal layers but not necessarily on all, which is called partial pilot-based channel estimation, as illustrated in Figure \ref{fig:channel_estimation_illustration}.}
\begin{figure}
\centering
	\includegraphics[width=0.4\textwidth]{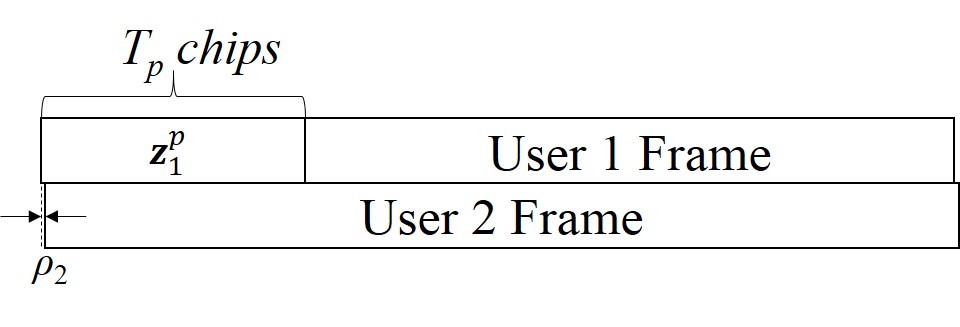}
    \caption{\label{fig:channel_estimation_illustration}Illustration of partial pilot-based channel estimation for $L = 2, L_p = 1$.}
\end{figure}
Without loss of generality, we assume to transmit pilot sequences $\boldsymbol{Z}^p = [\bz^p_1, \bz^p_2, \dots, \bz^p_{L_p}]$ in $L_p$ layers, where $\bz^p_i$ denotes the pilot sequence in layer $i$ for $1 \leq i \leq L_p$ and $0 \leq L_p < L$.
Let ${\bS}^p_1, {\bS}^p_2, \ldots, {\bS}^p_{T_p}$ denote the state sequence for channel estimation, where $T_p$ denotes the number of chips.
We have that ${\bS}^p_{t} = [z^p_{1, \lceil \frac{t}{L} \rceil},  \dots, z^p_{L_p, \lceil \frac{t-1}{L} \rceil}, z_{L_p+1, \lceil \frac{t-1}{L} \rceil}, \dots, z_{L, \lceil \frac{t-L+1}{L} \rceil}]^T$, and $\hat{\boldsymbol{\Lambda}}$ is estimated based on EM algorithm.
Let $\bN^p = [ N^p_1, N^p_2, \ldots, N^p_{T_p}]$ denote the number of received photoelectrons in each chip for channel estimation, where $\bN^p$ is the observation sequence of ${\bS}^p_t$ for $1 \leq t \leq T_p$.
The estimation for $\hat{\boldsymbol{\Lambda}}$ is processed by $V$ iterations, and in each iteration the updating rule is provided as follows.

\textbf{E-step}: In the $v$th iteration, based on the estimate result $\hat{\lambda}^{(v-1)}_{\boldsymbol{s}_i}$ in the $(v-1)$th iteration, the a posterior probability of $\boldsymbol{S}^p_t$ is given by
\begin{equation}
\begin{aligned}
Q^{(v)}(\boldsymbol{S}^p_t = \boldsymbol{s}_i)
&= \mathbb{P} (\boldsymbol{S}^p_t = \boldsymbol{s}_i | \bN^p, \lambda_{\boldsymbol{s}_i} = \hat{\lambda}^{(v-1)}_{\boldsymbol{s}_i})
\\
&= \frac{\mathbb{P} (\bN^p, \boldsymbol{S}^p_t = \boldsymbol{s}_i | \lambda_{\boldsymbol{s}_i} = \hat{\lambda}^{(v-1)}_{\boldsymbol{s}_i})} {\sum_{\boldsymbol{s}_i \in \mathscr{B}^{L \setminus L_p}} \mathbb{P} (\bN^p, \boldsymbol{S}^p_t = \boldsymbol{s}_i| \lambda_{\boldsymbol{s}_i} = \hat{\lambda}^{(v-1)}_{\boldsymbol{s}_i})},
\end{aligned}
\end{equation}
where $\mathscr{B}^{L \setminus L_p} = \big\{ \sum_{i=1}^{L_p} z^p_{i, \lceil \frac{t-i+1}{L} \rceil} \boldsymbol{e}_i + \sum_{i=L_p+1}^L \theta_i \boldsymbol{e}_i \mid \theta_i \in \{ 0, 1\}, L_p+1 \leq i \leq L \big\}$, and
\begin{equation}
\begin{aligned}
\mathbb{P} (\bN^p, \boldsymbol{S}^p_t = \boldsymbol{s}_i | \lambda_{\boldsymbol{s}_i} = \hat{\lambda}^{(v-1)}_{\boldsymbol{s}_i}) = \frac{ \Big( \tau_t \hat{\lambda}^{(v-1)}_{\boldsymbol{s}_i} \Big)^{N^p_t} }{ N^p_{t} !} e^{-\tau_t \hat{\lambda}^{(v-1)}_{\boldsymbol{s}_i}}.
\end{aligned}
\end{equation}

\textbf{M-step}: Given a posterior probability $Q^{(v)}(\boldsymbol{S}^p_t = \boldsymbol{s}_i)$ for the $v$th iteration, the ML-estimation for $\hat{\boldsymbol{\Lambda}}^{(v)}_T = \{ \hat{\lambda}^{(v)}_{\boldsymbol{s}_i} | \boldsymbol{s}_i \in \mathscr{B}^{L \setminus L_p}\}$ is given by
\begin{equation}
\begin{aligned}
\label{eq:M_step}
\hat{\lambda}^{(v)}_{\boldsymbol{s}_j} = \frac{ \sum^{T_p}_{t=1} Q^{(v)}(\boldsymbol{S}^p_t = \boldsymbol{s}_i) N^p_{t} }{ \sum^{T_p}_{t=1} Q^{(v)}(\boldsymbol{S}^p_t = \boldsymbol{s}_i) \tau_t},
\end{aligned}
\end{equation}
where the preset initial $\hat{\boldsymbol{\Lambda}}^{(0)}$ must satisfy $\big( \hat{\lambda}^{(0)}_{\boldsymbol{s}_i} - \hat{\lambda}^{(0)}_{\boldsymbol{s}_j} \big) \big( {\lambda}_{\boldsymbol{s}_i} - {\lambda}_{\boldsymbol{s}_j} \big) > 0$ for $i \neq j$ and ${\lambda}_{\boldsymbol{s}_i} \neq {\lambda}_{\boldsymbol{s}_j}$.
\begin{proof}
Please refer to Appendix E.
\end{proof}

\subsection{HMM-Based Symbol Detection}

Based on HMM, the receiver aims to detect state sequence $\varmathbb{S}_T$ according to the observation sequence $\boldsymbol{N}_T$ and ($\boldsymbol{\pi}_1$, $\boldsymbol{A}_t$, $\boldsymbol{B}_t$).
The trellis diagram for HMM is adopted to find the optimal state transition path maximizing the likelihood function or a posteriori probability.
Figure \ref{fig:gridgraph2} illustrates the trellis diagram for $L = 3$, where each state $\boldsymbol{S}_t$ is expressed as $\{z_{k,\lceil \frac{t-k+1}{L} \rceil} | 1 \leq k \leq L\}$, and each branch between adjacent states corresponds to a non-zero element of $\boldsymbol{A}_t$.
We adopt Viterbi and Bahl-Cocke-Jelinek-Raviv (BCJR) algorithms to maximize the likelihood function $\mathbb{P} (\boldsymbol{N}_T| \varmathbb{S}_T = \boldsymbol{s}_T )$ and a posteriori probability $\mathbb{P} (\varmathbb{Z}_T = \boldsymbol{z}_T | \boldsymbol{N}_T)$, respectively, and minimizes the error rate of sequence and symbol detection, respectively.
\begin{figure}
\centering
	\includegraphics[width=1.0\textwidth]{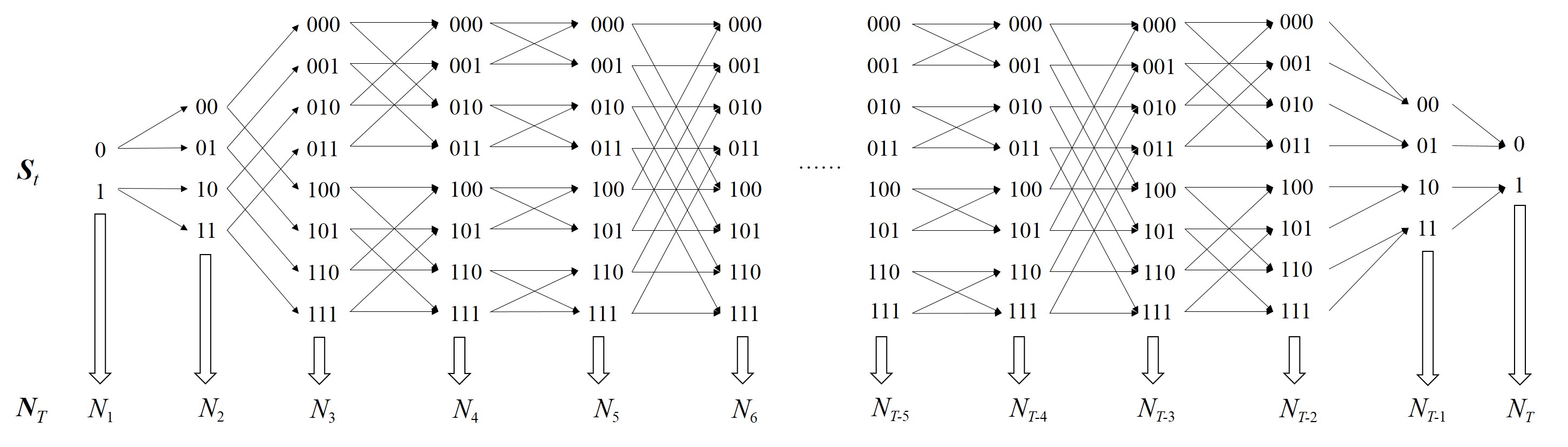}
    \caption{\label{fig:gridgraph2}The trellis diagram for $L = 3$.}
\end{figure}


For Viterbi algorithm, we maximize the log-likelihood function of state sequence summarized as follows
\begin{equation}
\begin{aligned}
\hat{\varmathbb{S}}_T = \arg \max \log \mathbb{P} (\boldsymbol{N}_T| \varmathbb{S}_T = \boldsymbol{s}_T ) = \arg \max \limits_{\boldsymbol{S}_t \in \mathscr{B}^{L}} \sum_{t = 1}^T N_t \log \tau_t \lambda_{\boldsymbol{S}_t} - \tau_t \lambda_{\boldsymbol{S}_t},
\end{aligned}
\end{equation}
where $\lambda_{\boldsymbol{S}_t}$ for $\boldsymbol{S}_t \in \mathscr{B}^{L}$ can be obtained from channel estimation.

Letting $\mathcal{L} (N_t | \boldsymbol{S}_t) = N_t \log \tau_t \lambda_{\boldsymbol{S}_t} - \tau_t \lambda_{\boldsymbol{S}_t}$, we have that $\mathcal{L} (\boldsymbol{N}_t | \varmathbb{S}_t) = \mathcal{L} (N_t | \boldsymbol{S}_t) + \mathcal{L} (\boldsymbol{N}_{t-1} | \varmathbb{S}_{t-1})$ for $2 \leq t \leq T$.
Thus dynamic programming is adopted with the following updated equation
\begin{equation}
\begin{aligned}
\max \mathcal{L} (\boldsymbol{N}_{t+1} | \varmathbb{S}_t, \boldsymbol{S}_{t+1,j})
= \mathcal{L} (N_{t+1} | \boldsymbol{S}_{t+1,j}) + \max \limits_{a_{i,j,t} \neq 0} \big\{ \mathcal{L} (\boldsymbol{N}_{t} | \varmathbb{S}^{t-1}, \boldsymbol{S}_{t,i}) \big\},
\end{aligned}
\end{equation}
which is initialized by $\mathcal{L} (\boldsymbol{N}_1 | \varmathbb{S}_1) = \mathcal{L} (N_1 | \boldsymbol{S}_1) \sim N_1 \log \tau_1(\lambda_0 + \lambda_{\boldsymbol{S}_1}) -\tau_1 (\lambda_0 + \lambda_{\boldsymbol{S}_1})$.
The detected symbol sequence can be retrieved via tracing back the optimal path.

For BCJR Algorithm, we maximize the posterior probability for each symbol $z_{k,i}$ for $1 \leq k \leq L$ and $1 \leq i \leq M$ as follows
\begin{equation}
\begin{aligned}
\label{eq:bcjr_max}
\hat{z}_{i, j} &= \arg \max \log \mathbb{P} (z_{k,i} | \boldsymbol{N}_T)
\\
&= \arg \max \log \mathbb{P} (\varmathbb{S}_{z_{k,i}} | \boldsymbol{N}_T)
\\
&\sim \arg \max \log \mathbb{P} (\varmathbb{S}_{z_{k,i}}, \boldsymbol{N}_T),
\end{aligned}
\end{equation}
where $\varmathbb{S}_{z_{k,i}} = \{ \bS_t | t = (i-1)L+k, (i-1)L+k+1, \dots, iL+k-1 \}$.

To obtain $\mathbb{P} (\varmathbb{S}_{z_{k,i}}, \boldsymbol{N}_T)$, we define the following probability functions
\begin{equation}
\begin{aligned}
\alpha_t(\boldsymbol{s}) &= \mathbb{P}(\boldsymbol{S}_t = \boldsymbol{s}, \boldsymbol{N}_t),
\\
\beta_t(\boldsymbol{s}) &= \mathbb{P}(\boldsymbol{N}_{[t+1,T]} | \boldsymbol{S}_t = \boldsymbol{s}),
\\
\gamma_t(\boldsymbol{v}, \boldsymbol{s}) &= \mathbb{P}(N_t, \boldsymbol{S}_t = \boldsymbol{s} | \boldsymbol{S}_{t-1} = \boldsymbol{v}),
\end{aligned}
\end{equation}
where $\boldsymbol{N}_{[a,b]} = \{ N_t | a \leq t \leq b \}$.
Note that we have
\begin{equation}
\begin{aligned}
&\mathbb{P} (\varmathbb{S}_{z_{k,i}} = \boldsymbol{\$}_{k,i}, \boldsymbol{N}_T) = \alpha_{(i-1)L+k}(\boldsymbol{s}_{(i-1)L+k})
\beta_{iL+k-1}(\boldsymbol{s}_{iL+k-1})
\!\! \prod^{iL+k-2}_{t=(i-1)L+k} \!\! \gamma_t(\boldsymbol{s}_{t}, \boldsymbol{s}_{t+1}),
\end{aligned}
\end{equation}
where $\boldsymbol{\$}_{k,i} = \{ \boldsymbol{s}_t | (i-1)L+k \leq t \leq iL+k-1 \}$.
Furthermore, we have that  $\gamma_t(\boldsymbol{s}_{t-1,i}, \boldsymbol{s}_{t,j}, n) = a_{t-1, i, j} b_{t, j,n+1}$ for $\boldsymbol{s}_{t-1,i}, \boldsymbol{s}_{t,j} \in \mathscr{B}^{L}$ and $n \in \mathbb{N}$.
Then, the calculations of $\alpha(\boldsymbol{s}_{t})$ and $\beta(\boldsymbol{s}_{t})$ are conducted according to the following recursive equations
\begin{equation}
\begin{aligned}
\alpha_t(\boldsymbol{s}_t) &= \sum_{\boldsymbol{s}_{t-1}\in\mathscr{B}^{L}} \alpha_{t-1}(\boldsymbol{s}_{t-1}) \gamma_t(\boldsymbol{s}_{t-1}, \boldsymbol{s}_{t}),
\\
\beta_t(\boldsymbol{s}_t) &= \sum_{\boldsymbol{s}_{t+1}\in\mathscr{B}^{L}} \beta_{t+1}(\boldsymbol{s}_{t+1}) \gamma_{t+1}(\boldsymbol{s}_{t}, \boldsymbol{s}_{t+1}).
\end{aligned}
\end{equation}
The initial values are $\alpha_1(\boldsymbol{s}_{1,i}) = \boldsymbol{\pi}_1 (\boldsymbol{s}_{1,i}) b_{i,1,N_1+1}$ for $\boldsymbol{s}_{1,i} \in \mathscr{B}^L$ and
$\beta_T(\boldsymbol{s}_{T}) = \boldsymbol{\pi}_{T} (\boldsymbol{s}_{T})$ for $\boldsymbol{s}_T \in \mathscr{B}^L$, where $\boldsymbol{\pi}_1$ is given by Equation (\ref{eq:HMM_pi1}); and $\boldsymbol{\pi}_t (\boldsymbol{s}_{t+1,i}) = \mathbb{P} (\boldsymbol{S}_{t+1} = \boldsymbol{s}_{t+1,i})$ can be obtained by the following recursive equation,
\begin{equation}
\begin{aligned}
\boldsymbol{\pi}_t (\boldsymbol{s}_{t+1,i}) = \sum_{\boldsymbol{s}_{{t},j} \in \mathscr{B}^{L}} a_{t,j,i} \boldsymbol{\pi}_{t-1} (\boldsymbol{s}_{t,j}).
\end{aligned}
\end{equation}

\subsection{Joint Detection and Decoding}

We adopt joint detection and decoding based on turbo processing.
For ML and MAP decoding, the log-likelihood ratio ($LLR$) and log-aposterior ratio ($LAR$) are adopted as the input soft information to the soft channel decoder, respectively.
Let $LLR^{(v)}_{z_{k,i}}$ and $LAR^{(v)}_{z_{k,i}}$ denote the log-likelihood ratio and log-aposterior-ratio of $z_{k,i}$ after the $v$-th iteration, respectively.
Typically each iteration of the turbo processing consists of one ML/MAP symbol detection operation followed by $V$ channel decoding iterations.

For the ML-decoding, the initial $LLR$ values are obtained by Viterbi algorithm as follows,
\begin{equation}
\begin{aligned}
LLR^{(0)}_{z_{k,i}} = \log \frac{\mathbb{P} (\boldsymbol{N}_T | z_{k,i} = 1)}{\mathbb{P} (\boldsymbol{N}_T | z_{k,i} = 0)} = \sum_{t=(i-1)L+k}^{iL+k-1} \log \frac{ \mathbb{P} (N_t | z_{k,i} = 1)} {\mathbb{P} (N_t | z_{k,i} = 0)},
\end{aligned}
\end{equation}
and the $LLR$ of the $i$-th transmitted symbol in layer $k$ in the $v$-th iteration is calculated by
\begin{equation}
\begin{aligned}
\label{eq:LLR_ML}
LLR^{(v)}_{z_{k,i}} = \log \frac{\mathbb{P} (\bN_T | z_{k,i} = 1)}{\mathbb{P} (\bN_T | z_{k,i} = 0)} = \!\! \sum_{t=(i-1)L+k}^{iL+k-1} \!\! \log \frac{ 
\mathbb{E}_{z_{k,i}=1}\mathbb{P} (N_t | \boldsymbol{S}_t = \boldsymbol{s}_t) } 
{\mathbb{E}_{z_{k,i}=0} \mathbb{P} (N_t | \boldsymbol{S}_t = \boldsymbol{s}_t) },
\end{aligned}
\end{equation}
where
the expectation $\mathbb{E}_{z_{k,i} = \theta} [\bullet]$ for $\theta \in \{ 0, 1 \}$ is calculated based on a posterior probabilities by the $(v-1)$th iteration of channels in $\mathcal{L} \setminus k$ as follows
\begin{equation}
\begin{aligned}
\label{eq:expection_ML}
\mathbb{E}_{\boldsymbol{s}_t\in\mathcal{S}_{z_{k,i}=\theta}} \![\bullet] = \!\!\! \sum_{\boldsymbol{s}_t\in\mathcal{S}_{z_{k,i}=\theta}} \prod_{j \in \mathcal{L} \setminus k} \!\! {\mathbb{P}}^{1-z_{j,\lceil \frac{t-j+1}{L} \rceil}}_{(v-1)} (z_{j,\lceil \frac{t-j+1}{L} \rceil} = 0 | N_t) {\mathbb{P}}^{z_{j, \lceil \frac{t-j+1}{L} \rceil}}_{(v-1)} (z_{j,\lceil \frac{t-j+1}{L} \rceil} = 1 | N_t) [\bullet],
\end{aligned}
\end{equation}
where $\mathcal{S}_{z_{k,i}=\theta} = \{ \boldsymbol{s}_{t} | \boldsymbol{e}^T_k \cdot \boldsymbol{s}_{t} = \theta \}$; $z_{j,\lceil \frac{t-j+1}{L} \rceil} = \boldsymbol{e}^T_{j} \cdot \boldsymbol{s}_{t}$; and the a posterior probability of $z_{j,\lceil \frac{t-j+1}{L} \rceil}$ after the $(v-1)$-th iteration is given by
\begin{equation}
\begin{aligned}
{\mathbb{P}}_{(v-1)} (z_{j,\lceil \frac{t-j+1}{L} \rceil} = 0 | N_t) \!=\! 1 \!-\! {\mathbb{P}}_{(v-1)} (z_{j,\lceil \frac{t-j+1}{L} \rceil} = 1 | N_t) \!=\! \frac{1}{1 + \exp \Big( \frac{q_{j,\lceil \frac{t-j+1}{L} \rceil}}{1-q_{j,\lceil \frac{t-j+1}{L} \rceil}}
LLR^{(v-1)}_{j,\lceil \frac{t-j+1}{L} \rceil}
\Big) }.
\end{aligned}
\end{equation}

For MAP-decoding, the initial $LAR$ is determined by BLJR detection as follows
\begin{equation}
\begin{aligned}
LAR^{(0)}_{z_{k,i}} = \log \frac{\mathbb{P} (z_{k,i} = 1 | \bN_T)}{\mathbb{P} (z_{k,i} = 0 | \bN_T)};
\end{aligned}
\end{equation}
and the $LAR$ of symbol $z_{k,i}$ from the $v$-th iteration is given by
\begin{equation}
\begin{aligned}
LAR^{(v)}_{z_{k,i}}
&= \log \frac{\mathbb{P}_{(v-1)} (z_{j,\lceil \frac{t-j+1}{L} \rceil} = 1 | \bN_T)}{\mathbb{P}_{(v-1)} (z_{j,\lceil \frac{t-j+1}{L} \rceil} = 0 | \bN_T)}
\\
&= \log \frac{\mathbb{P} (\bN_T | z_{k,i} = 1)}{\mathbb{P} (\bN_T | z_{k,i} = 0)} + \log \frac{ \mathbb{P} (z_{k,i} = 1) }{\mathbb{P} (z_{k,i} = 0)}
\\
&= LLR^{(v-1)}_{z_{k,i}} + \log \frac{ q_{k,i} }{1 - q_{k,i}},
\end{aligned}
\end{equation}
where $LLR^{(v-1)}_{z_{k,i}}$ is computed according to Equations (\ref{eq:LLR_ML}) and (\ref{eq:expection_ML}).
Furthermore, the a posterior probability of MAP-decoding is given by
\begin{equation}
\begin{aligned}
{\mathbb{P}}_{(v-1)} (z_{j,\lceil \frac{t-j+1}{L} \rceil} = 0 | N_t) \!=\! 1 \!-\! {\mathbb{P}}_{(v-1)} (z_{j,\lceil \frac{t-j+1}{L} \rceil} = 1 | N_t) \!=\! \frac{1}{1 + \exp \Big( LAR^{(v-1)}_{j,\lceil \frac{t-j+1}{L} \rceil} \Big)}.
\end{aligned}
\end{equation}

\section{Numerical and Sumulation Results}

In this section, we provide numerical and simulation results on the achievable rates, power allocation, channel estimation as well as joint detection and decoding.

\subsection{Achievable Rates}

Consider the superimposed transmission with $L=2$ signal layers, where $\lambda_1 = \lambda_2 = 10$ and background radiation $\lambda_0 = 0.01$.
We evaluate the sum achievable transmission rate versus symbol number $M$ and relative delay $\rho_1$ in Figure \ref{fig:sum_rate_vs_T},
\begin{figure}
\centering
	\includegraphics[width=0.7\textwidth]{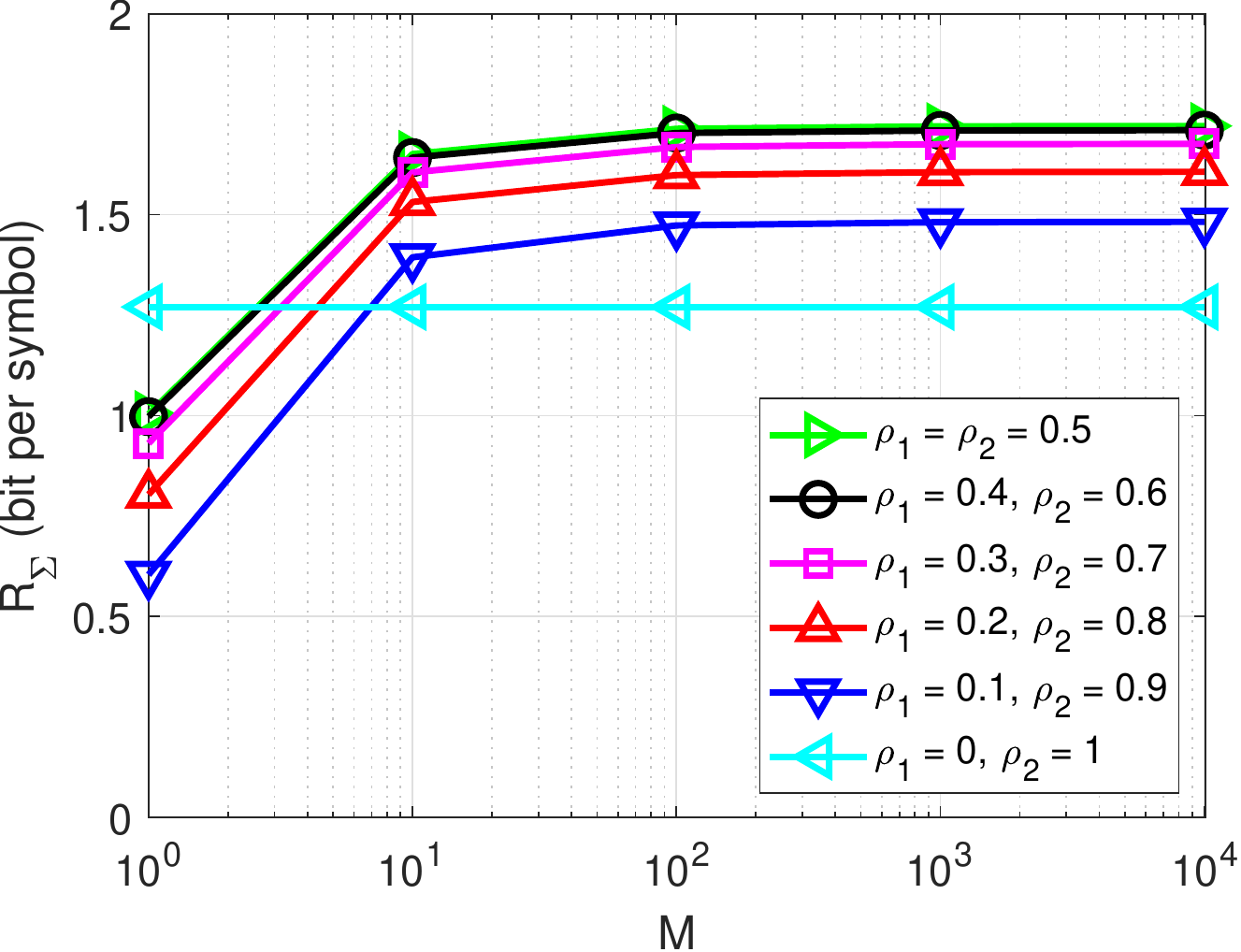}
    \caption{\label{fig:sum_rate_vs_T} The achievable sum rates with different relative delays.}
\end{figure}
where the scenario of $\rho_1 = 0, \rho_2 = 1$ for perfect symbol boundary alignment is also shown for comparison.
It is implied that introducing relative relays can enhance the achievable sum rate, and $\rho_1 = 0.5$ can maximize the sum rate, which can converge for $M$ exceeding $10^2$, where an improvement of $0.5$ bit per symbol can be observed.


Consider a more general scenario with possibly more than $2$ signal layers, i.e., $M = 1\times10^4$, $\lambda_i = \lambda_j = \lambda$ for $1 \leq i < j \leq L$ and background radiation $\lambda_0 = 0.01$.
The achievable sum rates for the case of $L = 2,3,4$ with the relative delays are shown in Figure \ref{fig:sum_rate_vs_Channel}, where the scenario of $L = 1$ without signal superposition is also shown for comparison.
\begin{figure}
\centering
	\includegraphics[width=0.7\textwidth]{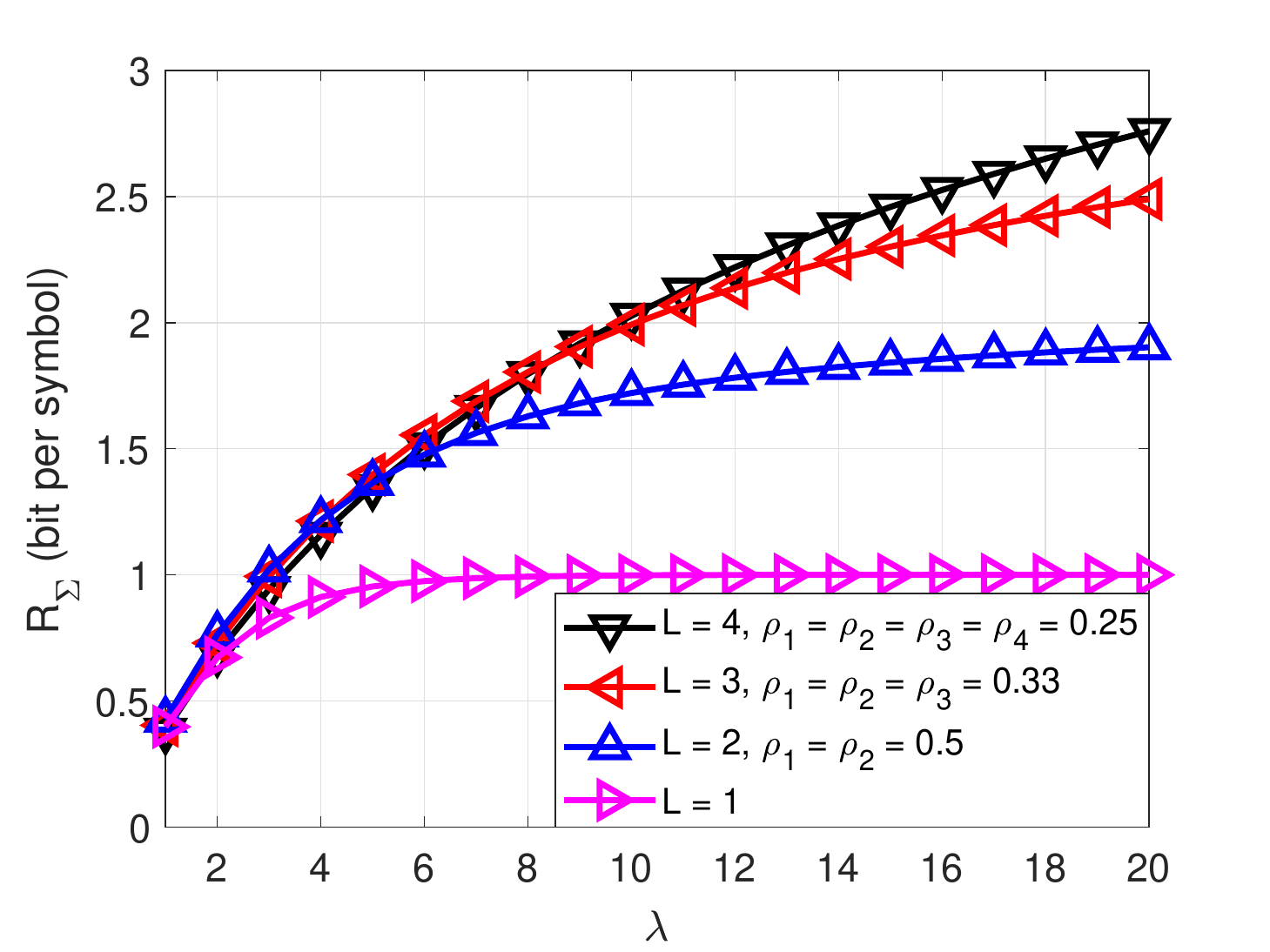}
    \caption{\label{fig:sum_rate_vs_Channel} The achievable sum rates of 2,3,4 signal layers.}
\end{figure}
It is seen that the achievable sum rate can be improved with sufficient receiver-side signal intensity.
Since the computational complexity of symbol detection grows exponentially with $L$, we can set a standard on the minimum $L$ subject to at least $\sigma$ bit per symbol gain over $L-1$ signal layers.
Accordingly, we can achieve the optimal number of signal layers $L^{*}$ corresponding to different $\lambda$.
For example $\sigma = 0.2$; when $\lambda < 3$, $L = 1$ is optimal; when $3 \leq \lambda \leq 8$, $L = 2$; when $8 \leq \lambda \leq 18$, $L = 3$; and when $\lambda > 18$, $L = 4$.

\subsection{Power Allocation}

We consider the power allocation in Section III.D for $L = 2$.
The first optimization problem is $\max \big(R^{*}_1 + R^{*}_2\big)$, subject to $\lambda_1 + \lambda_2 = \lambda_s$.
Figure \ref{fig:Power_alloc} plots the maximum achievable sum rates and their optimal power allocation versus $\lambda_s$.
It is seen that as $\lambda_s$ increases, the optimal power allocation tends to become equal distribution, where the achievable sum rate enhances as $\rho$ grows from $0.1$ to $0.5$.
The second optimization problem is to $\max R^*_{1}$, subjected to $R^*_{2} \geq R_{2inf}$ and $\lambda_1 + \lambda_2 = \lambda_s$, as shown in Figure \ref{fig:R_lambda_2u}.
Let $P$ denote the intersection of lines $R = R_{2inf}$ and $\lambda_1 + \lambda_2 = \lambda_s$; and $\lambda_{P}$ denote the $x$-coordinates of $P$.
The feasible solution for the problem is that $\lambda_1 = \lambda_s - \lambda_{P}, \lambda_2 = \lambda_{P}$.

\begin{figure}
\centering
	\includegraphics[width=0.7\textwidth]{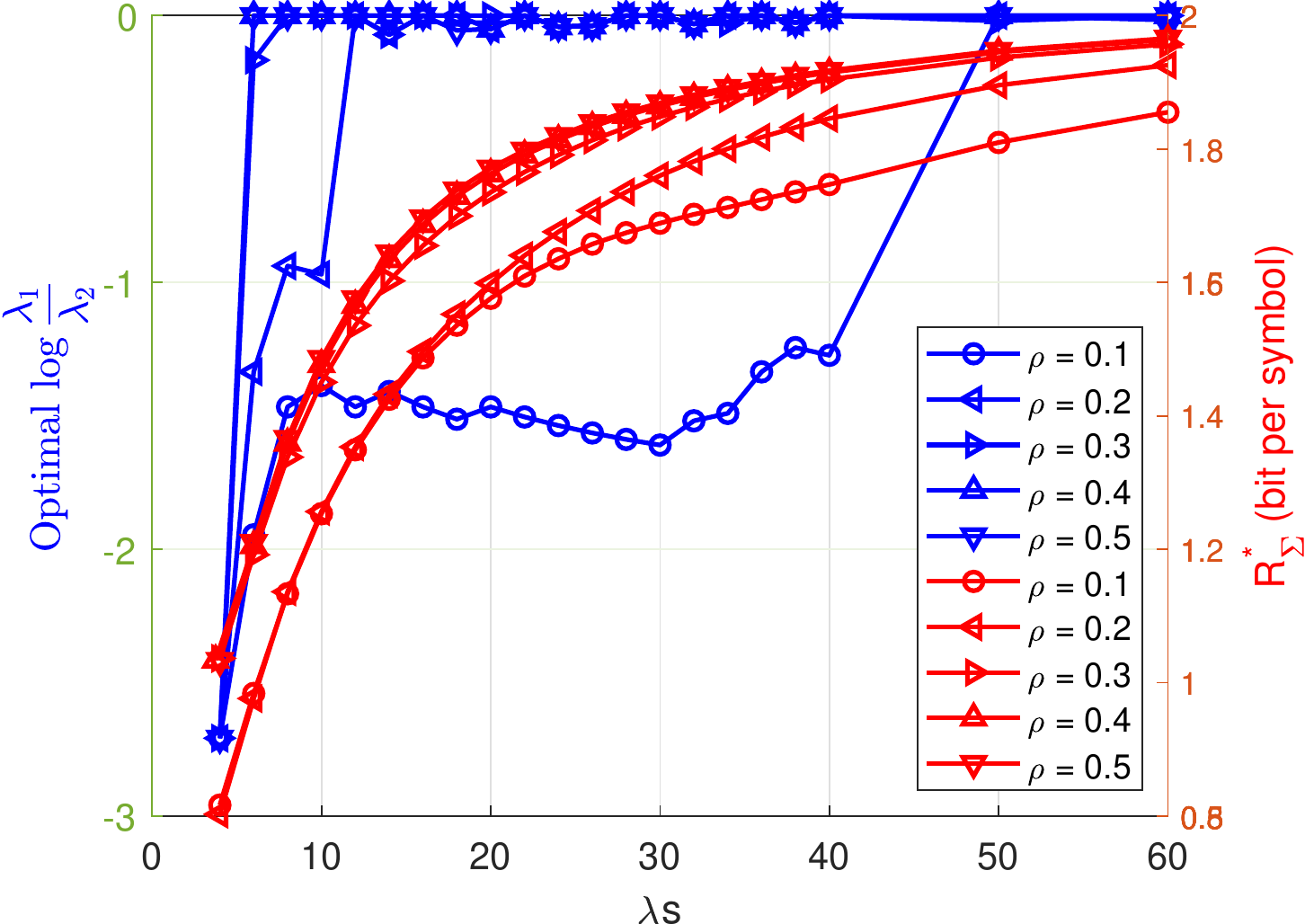}
    \caption{\label{fig:Power_alloc} The optimal power allocation and achievable sum rate versus $\lambda_s$.}
\end{figure}

\begin{figure}
\centering
	\includegraphics[width=0.7\textwidth]{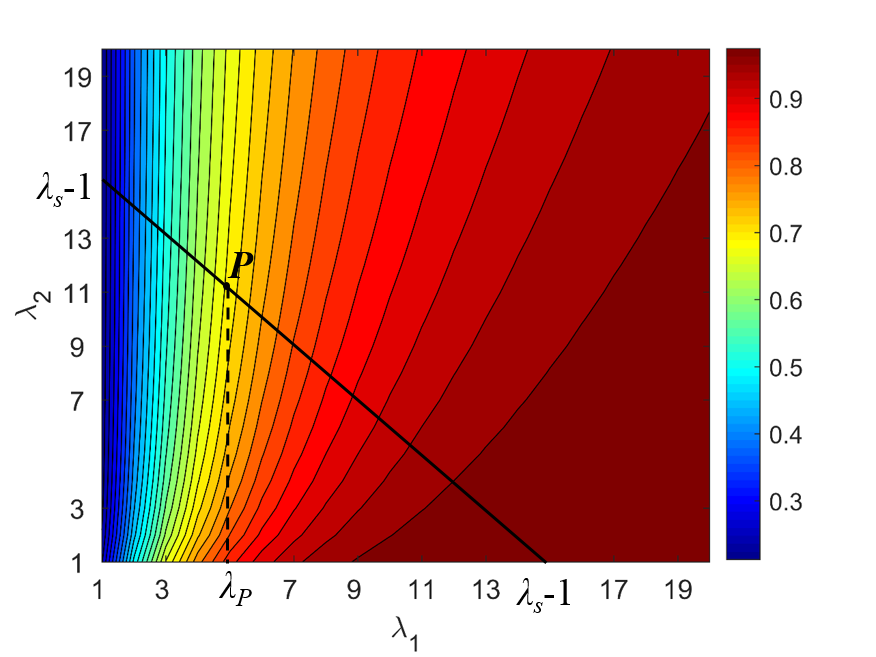}
    \caption{\label{fig:R_lambda_2u} The achievable transmission rate of a single layer.}
\end{figure}


\subsection{Joint Detection and Decoding}

Assume that $\lambda_i = \lambda_{ave}$ for $1 \leq i \leq L$.
The average symbol error rates for $L = 2$ and $L = 3$ of joint detection versus $\lambda_{ave}$ are illustrated in Figures \ref{fig:viterbi_BCJR_ber_2L} and \ref{fig:viterbi_BCJR_ber_3L}, respectively.
Furthermore, we adopt a $(12620, 6310)$ LDPC code for each signal layer, where the parity check matrix construction and low-complexity message pass decoding follow \cite{yang2004design, fossorier2004quasicyclic} and \cite{chen2005reduced}.
The average bit error rates for $L = 2$ and $L = 3$ by joint detection and decoding versus $\lambda_{ave}$ are shown in Figures \ref{fig:viterbi_bcjr_ldpc_ber_2L} and \ref{fig:viterbi_bcjr_ldpc_ber_3L}, respectively.
It is seen that for $L = 2$, $\rho = 0.5$ has the lowest error rate for both detection and decoding, which accords with the maximum achievable sum rate.

\begin{figure}
\subfigure[The number of signal layer $L$=2.]{
  \begin{minipage}[t]{0.48\linewidth}
    \centering
	\includegraphics[width=1.0\textwidth]{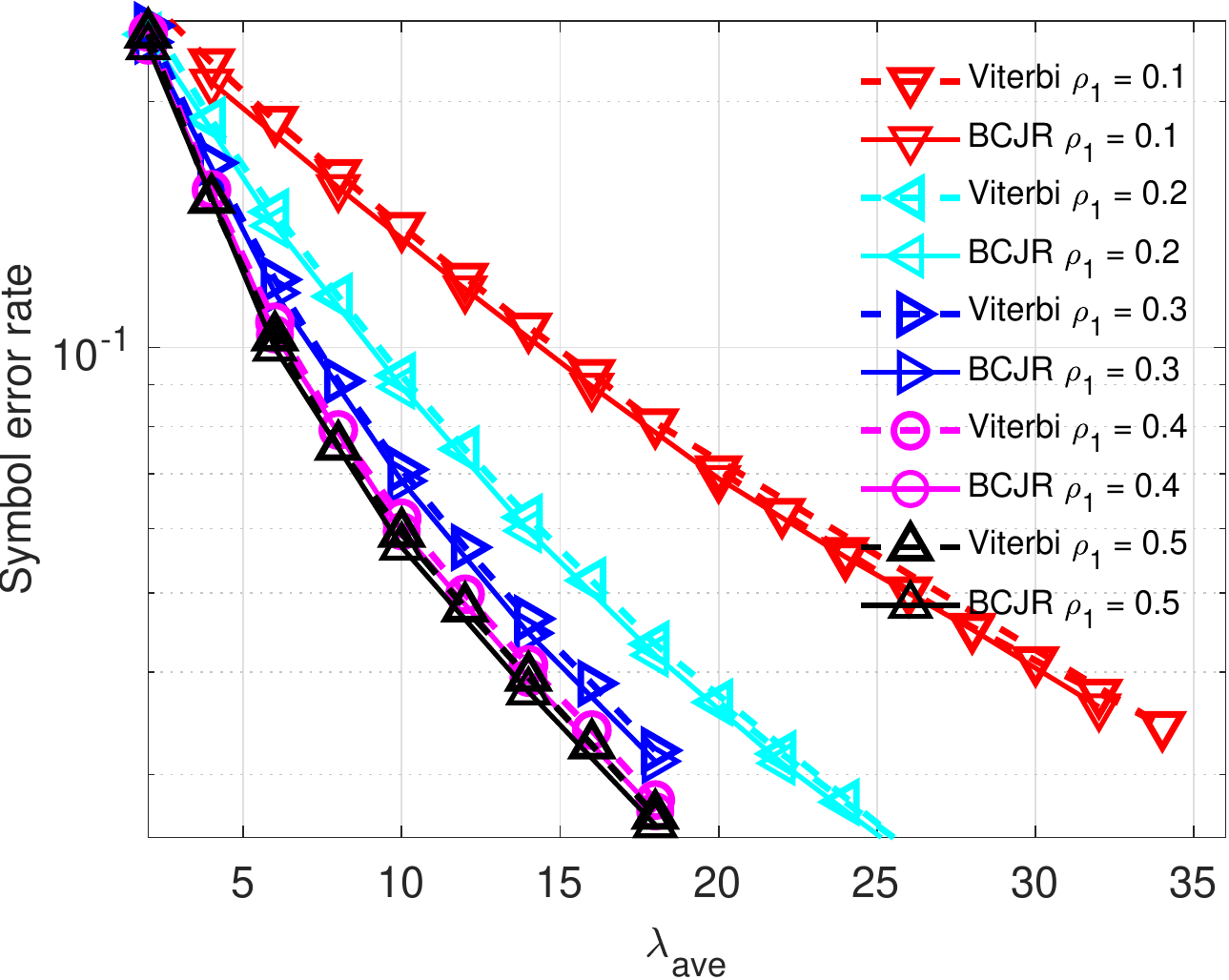}
    \label{fig:viterbi_BCJR_ber_2L}
  \end{minipage}%
}
\subfigure[The number of signal layer $L$=3.]{
  \begin{minipage}[t]{0.48\linewidth}
    \centering
    \includegraphics[width=1.0\textwidth]{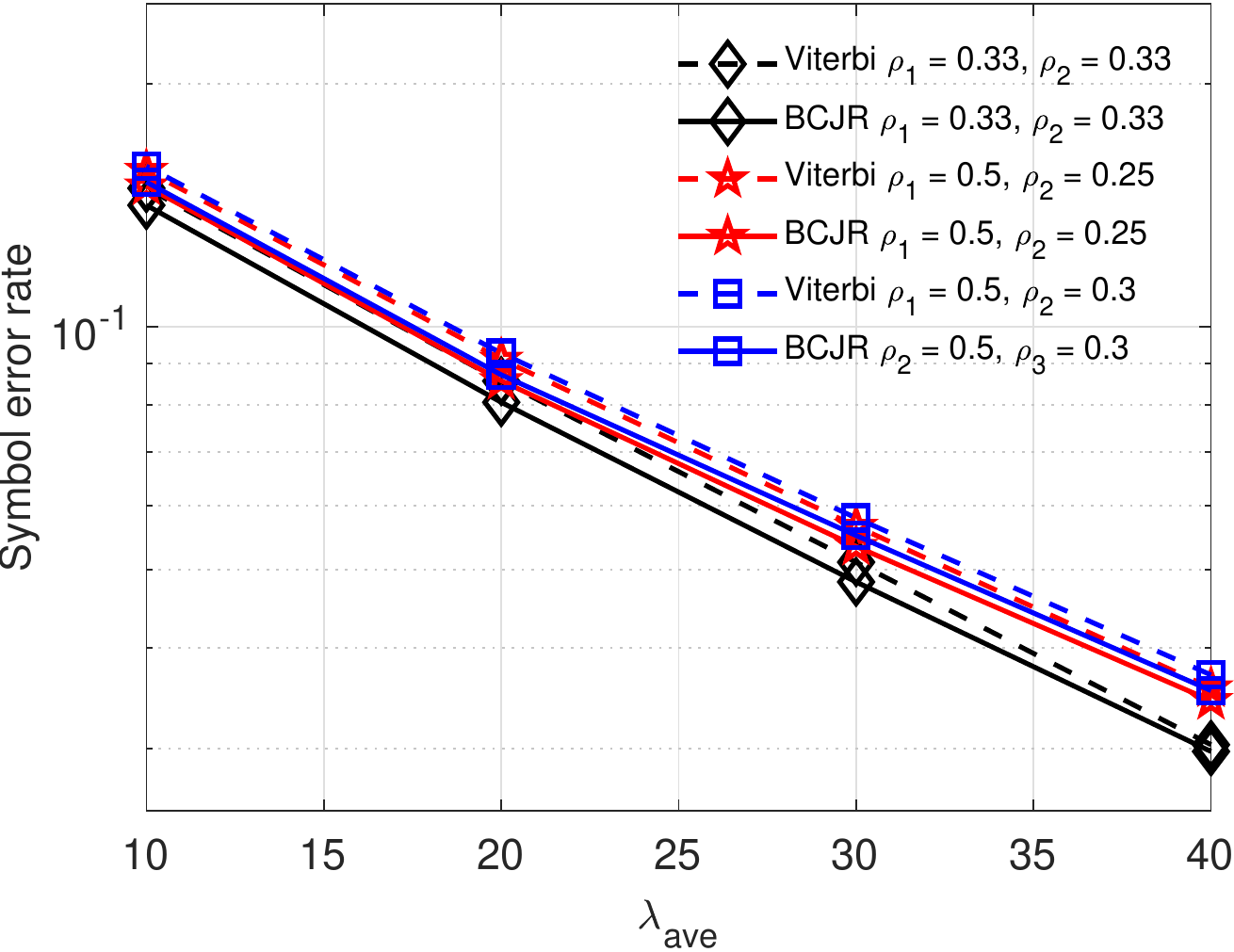}
    \label{fig:viterbi_BCJR_ber_3L}
  \end{minipage}
}
\caption{The symbol error rate of joint detection.}
\end{figure}

\begin{figure}
\subfigure[The number of signal layer $L$=2.]{
  \begin{minipage}[t]{0.48\linewidth}
    \centering
	\includegraphics[width=1.0\textwidth]{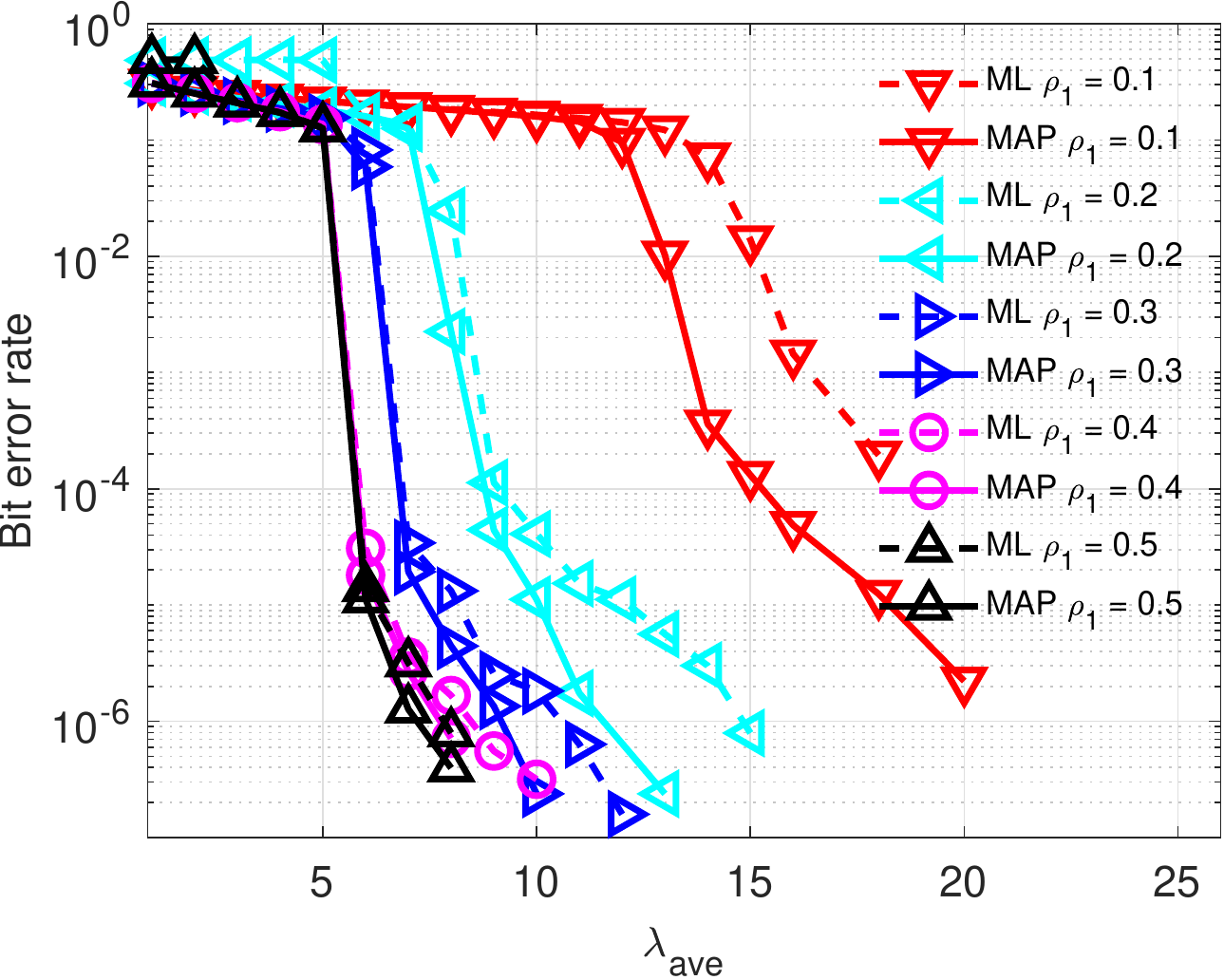}
    \label{fig:viterbi_bcjr_ldpc_ber_2L}
  \end{minipage}%
}
\subfigure[The number of signal layer $L$=3.]{
  \begin{minipage}[t]{0.48\linewidth}
    \centering
    \includegraphics[width=1.0\textwidth]{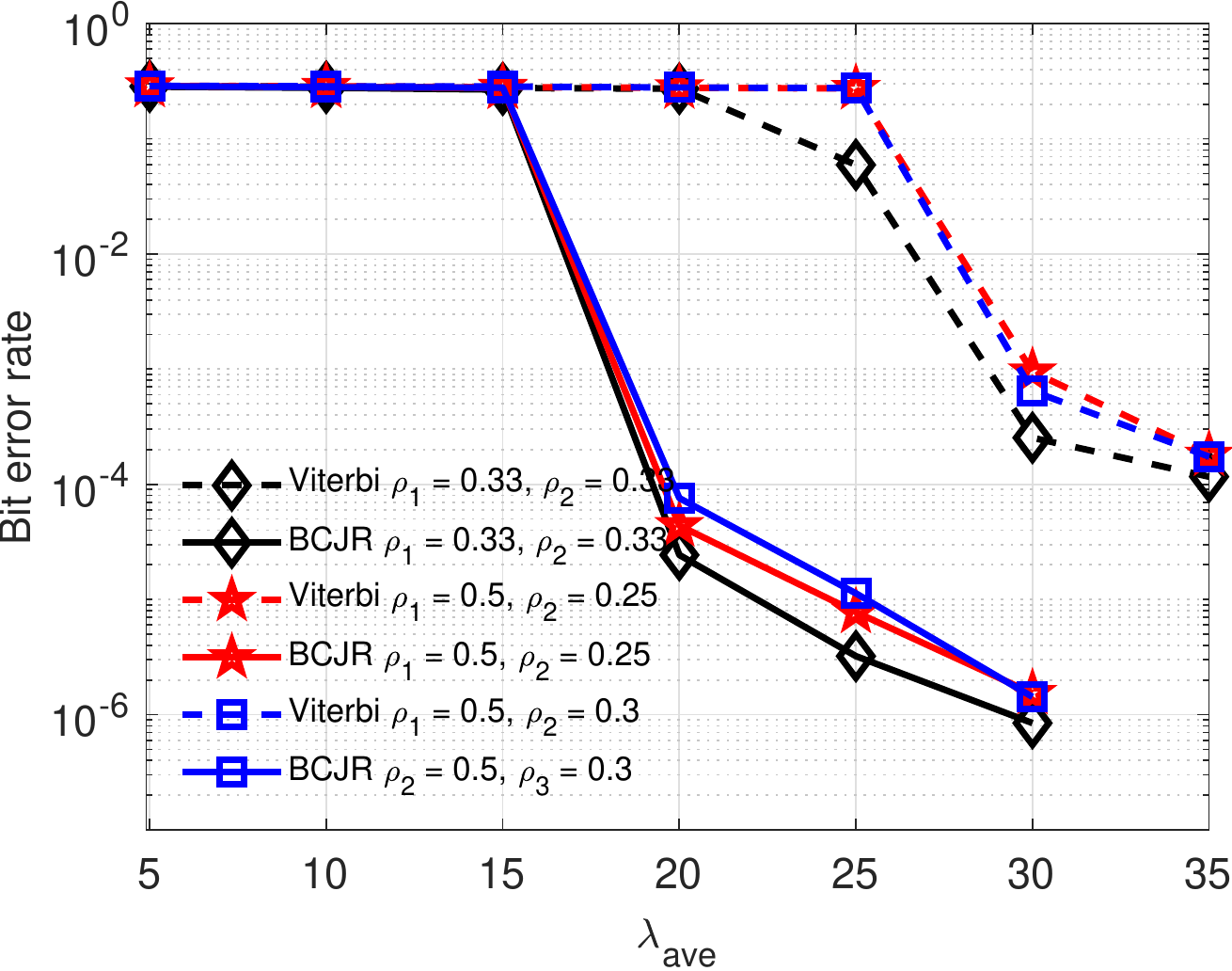}
    \label{fig:viterbi_bcjr_ldpc_ber_3L}
  \end{minipage}
}
\caption{The bit error rate of joint detection and decoding with $(12620, 6310)$ LDPC code.}
\end{figure}


\section{Experimental Results for $2$-layer-superimposed Transmission}

We conduct offline experiments on the $2$-layer-superposition transmission for optical wireless scattering communication to experimentally evaluate the proposed joint detection and decoding.
At the transmitter side, a waveform generator is adopted to produce OOK signals.
A Bias-Tee is employed to combine the AC and DC signals to drive the UV LED.
At the receiver side, a photomultiplier tube (PMT) is employed as the photon-detector, which is integrated with an optical filter in a sealed box.
The UV signal of wavelength around $280$nm can be detected, while the background radiation of other wavelengths is blocked.
The PMT output signal is attenuated by an attenuator, amplified by an amplifier, and then filtered by a low-pass filter, which is then sampled by the oscilloscope.
Finally, the photon counting processing, HMM-based MAP joint detection and decoding are realized in the received-side personal computer (PC) based on the sampled waveforms from the oscilloscope.
Table \ref{tb:experimental_parameters} shows the specification of experimental equipment, and Figures \ref{fig:ex_chart} and \ref{fig:TX-RX} illustrate the entire experimental block diagram and the test bed realizations, respectively.


\begin{figure}
\centering
	\includegraphics[width=0.8\textwidth]{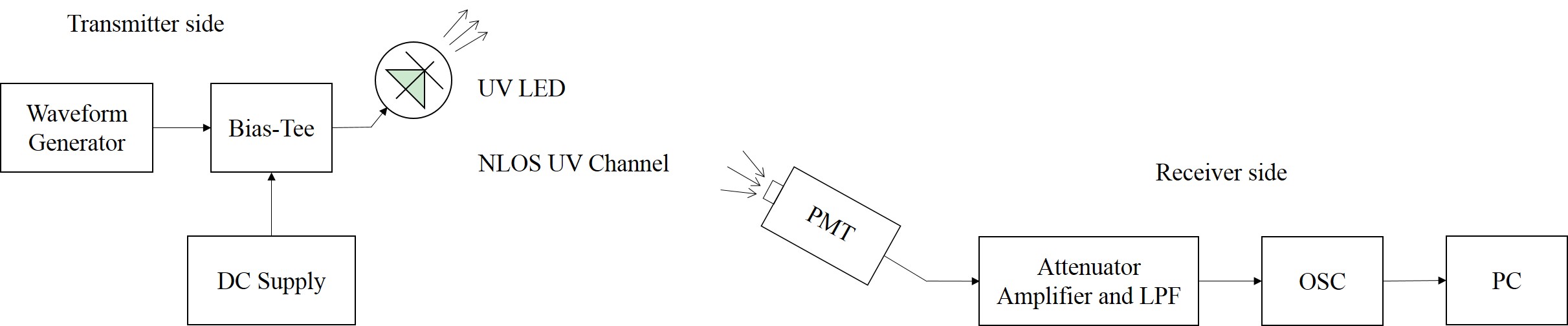}
    \caption{\label{fig:ex_chart} Diagram of the experimental superimposed communication system.}
\end{figure}

\begin{figure}
\centering
	\includegraphics[width=0.8\textwidth]{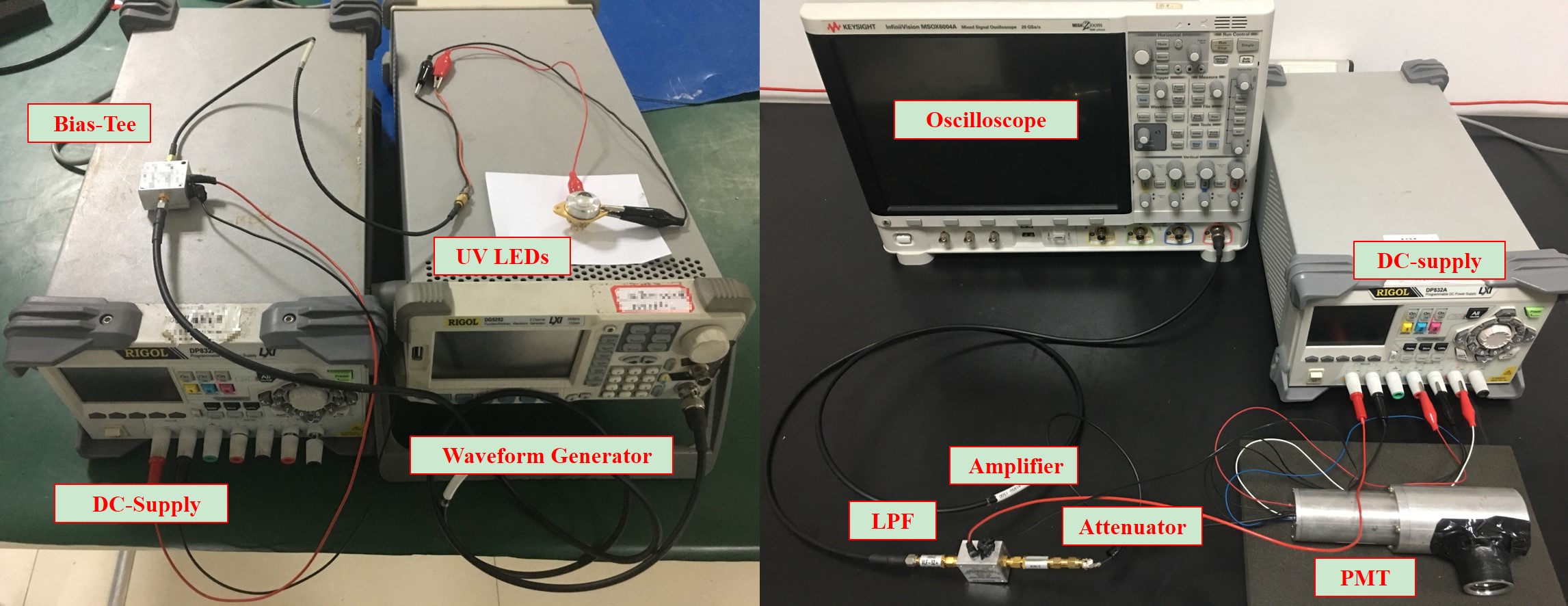}
    \caption{\label{fig:TX-RX} Demonstration of the transmitter-side (left) and receiver-side (right) test beds.}
\end{figure}

\begin{table}[tbp]
\caption{Specification of device for experiment.}
\label{tb:experimental_parameters}
\centering
\begin{tabular}{|c|c|c|}  
\hline
\multirow{2}{*}{UV LED}
&Model &TO-3zz PO\#2036 \\
\cline{2-3}
&Wavelength &$280$nm \\
\hline
\multirow{2}{*}{Optical filter}
&Peak transmission &$28.2$\% \\
\cline{2-3}
&Aperture size &$\Phi 31.5$mm $\times 28.3$mm \\
\hline
\multirow{4}{*}{PMT}
&Model &R7154 \\
\cline{2-3}
&Spectral response &$160$nm $\sim$ $320$nm \\
\cline{2-3}
&Dark counts &$< 10$ per second \\
\cline{2-3}
&Detection bandwidth &$> 200$MHz \\
\hline
\end{tabular}
\end{table}

In the experiment, the background radiation intensity is around $150$ photoelectrons per second in the indoor environment $(\lambda_0 \approx 1.5 \times 10^{-4})$.
Furthermore, we adopt the following parameters for two signal layers: symbol duration $T_s = 1 \mu s$; $\rho = 0.5$; uniform power allocation for $2$ signal layers $(\lambda_1 = \lambda_2 = \lambda_{ave})$; the same parity check matrix construction and decoding algorithm of LDPC codes as those in simulation; and the uniform prior probabilities for $0-1$ symbols.
For each $\lambda_{ave}$, we implement MAP joint detection and decoding and count the bit error rate based on the transmission of $1000$ frames $(1.262 \times 10^7$ random bits$)$.


We experimentally evaluate channel estimation for $L = 2, \rho_1 = 0.5$, where we exploit a $255$-bit m sequence as a pilot sequence $\boldsymbol{z}^p$.
For $L_p = 1$, $\boldsymbol{z}^p_1 = \boldsymbol{z}^p$; and for $L_p = 2$, $\boldsymbol{z}^p_1 = \boldsymbol{z}^p_2 = \boldsymbol{z}^p$.
The performance of channel estimation versus the number of iterations is illustrated in Figure \ref{fig:EX_estimation255}, where the result of $L_p = 2$ is from the ML estimation.
It is implied that real time estimation for both $L_p = 0$ and $1$ can converge to the ML solution; and assisted by the pilot sequence, the convergency of $L_p = 1$ is faster than $L_p = 0$.
Furthermore, higher $\lambda_{\boldsymbol{s}_i}$ with large receiver-side SNR can lead to faster convergence, which is close to the simulation result on the channel estimation with the same system parameters, as shown in Figure \ref{fig:SIMU_estimation255}.

\begin{figure}
\centering
\includegraphics[width=0.7\textwidth]{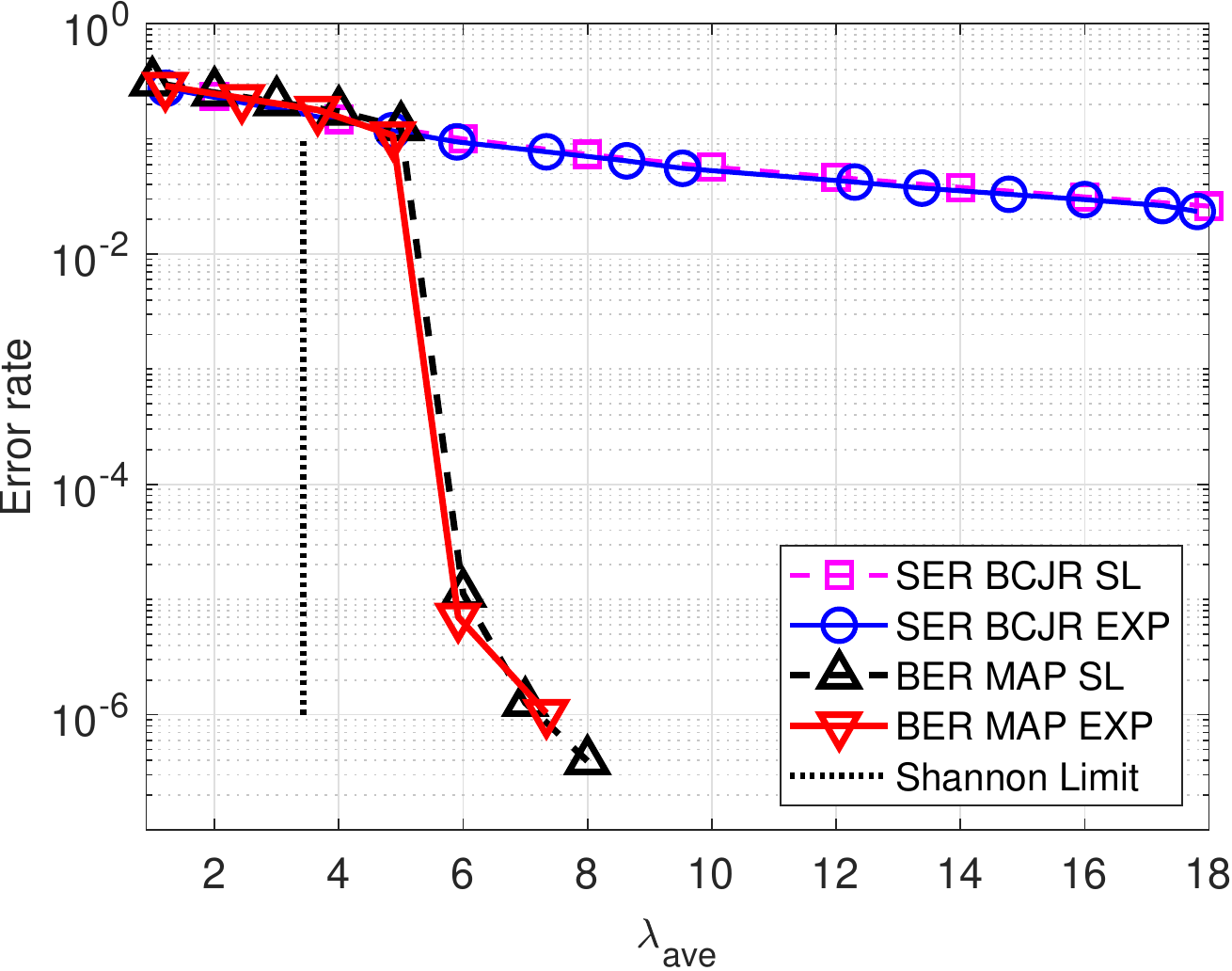}
    \caption{\label{fig:EX_detection} The average symbol and bit error rate of $2$-layer superimposed communication from simulation and experimental measurements.}
\end{figure}

\begin{figure}
\subfigure[Convergence of channel estimation from experiments.]{
  \begin{minipage}[t]{0.48\linewidth}
    \centering
    \includegraphics[width=1.0\textwidth]{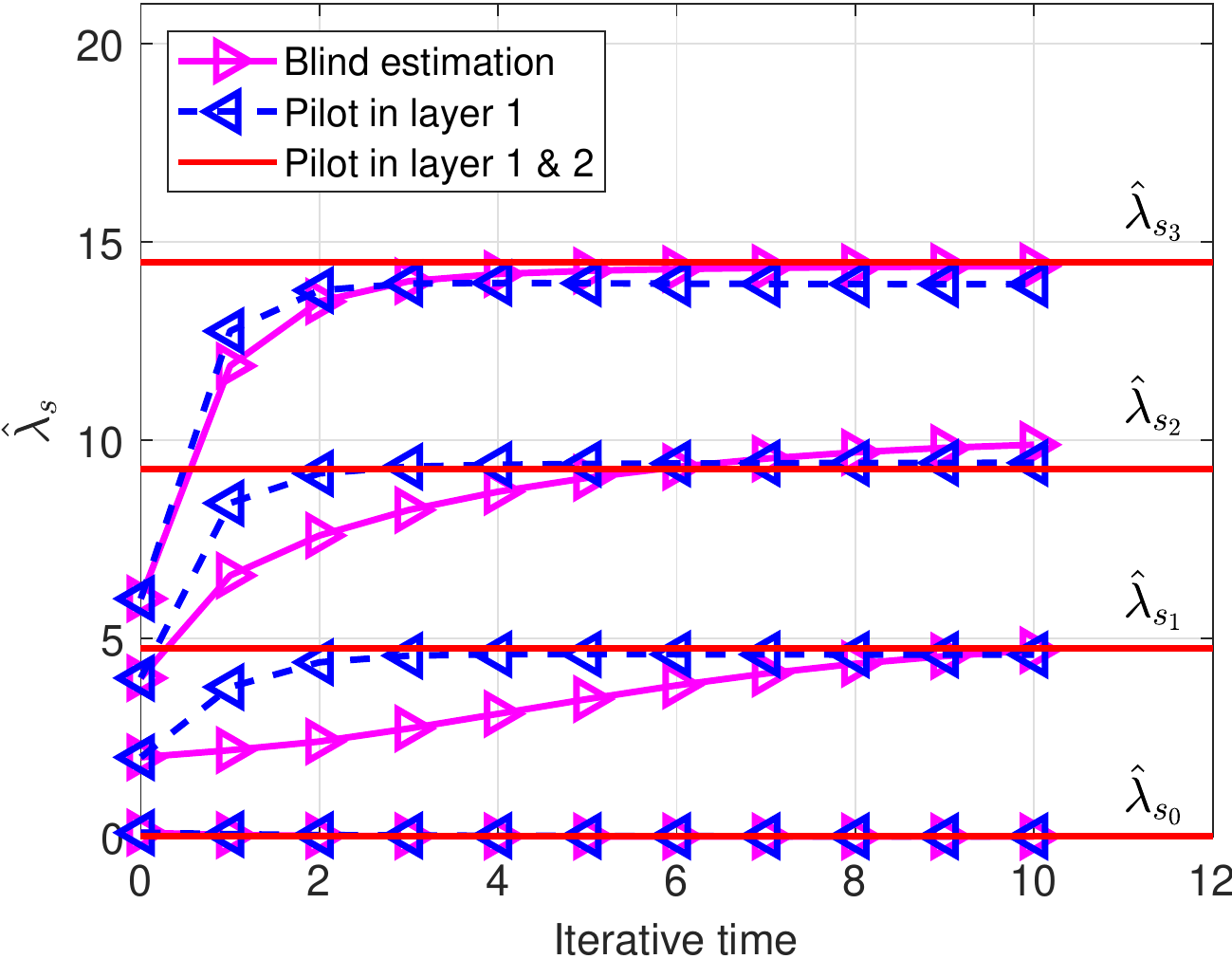}
    \label{fig:EX_estimation255}
  \end{minipage}
}
\subfigure[Convergence of channel estimation from simulations.]{
  \begin{minipage}[t]{0.48\linewidth}
    \centering
	\includegraphics[width=1.0\textwidth]{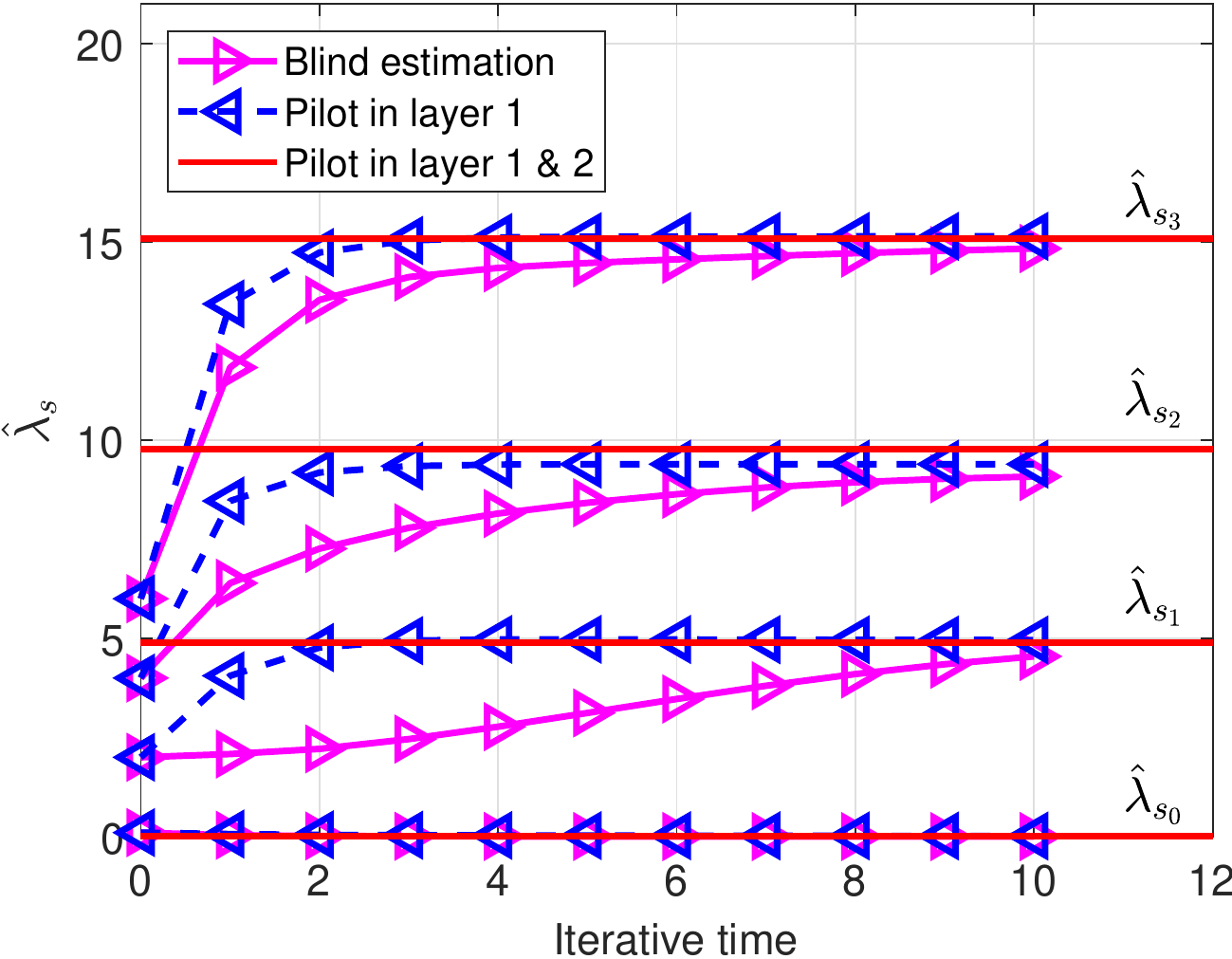}
    \label{fig:SIMU_estimation255}
  \end{minipage}%
  }
\caption{Convergence of channel estimation from both experiments and simulations.}
\end{figure}

Moreover, the MAP detection with and without LDPC code (denoted as EXP) is evaluated in Figure \ref{fig:EX_detection}, where the simulation results with the same channel parameters (denoted as SL) is plotted for comparison.
It is seen that the experimental results on the channel estimation, symbol detection and joint detection/decoding are close to the simulation results, which illustrates the feasibility of the proposed channel estimation and signal detection approaches in real communication scenarios.


\section{Conclusion}

We have proposed superposition transmission for optical wireless scattering communication based on HMM.
We have obtained the achievable rates of proposed superposition transmission, and proposed the EM-based channel estimation and joint detection and decoding.
The performance of the proposed approaches are verified by numerical results.
Moreover, for two- and three-layer transmission, both simulation and experimental results are employed to validate the feasibility of the proposed algorithms for channel estimation as well as joint detection and decoding.

\section{Appendix}

\textcolor{red}{
\subsection{Comparison of Achievable Rates between OOK and 2-Pulse-Position Modulations (2-PPM)}
}

\textcolor{red}{
The mutual information of single-use OOK modulation is given by
\begin{equation}
\begin{aligned}
\mathrm{I}_{OOK}(X; N) = \max\limits_{0<q<1} \Bigg\{ \mathcal{H}\bigg[ \sum_{X_q \in \{ 0,1\}}{\mathbb{P}(X_q) \mathbb{P}_{OOK}(N | X_q)} \bigg] -  \sum_{X_q \in \{ 0,1\}}{\mathbb{P}(X_q) \mathcal{H} \big[\mathbb{P}_{OOK}(N | X_q) \big]} \Bigg\},
\end{aligned}
\end{equation}
and that of 2-PPM is given by
\begin{equation}
\begin{aligned}
\mathrm{I}_{2-PPM}(X; N_1, N_2) = \max\limits_{0<q<1} \max\limits_{0<\tau<1} \Bigg\{ \mathcal{H}\bigg[ \sum_{X_q \in \{ 0,1\}}{\mathbb{P}(X_q) \mathbb{P}_{2-PPM}(N_1,N_2 | X_q, \tau) \bigg] } -  \sum_{X_q \in \{ 0,1\}}{\mathbb{P}(X_q) \mathcal{H} \big[ \mathbb{P}_{2-PPM}(N_1,N_2 | X_q, \tau) \big] } \Bigg\},
\end{aligned}
\end{equation}
where $X_q \sim \{ q, 1-q \}$;
\begin{equation}
\begin{aligned}
\mathbb{P}_{OOK}(N | X_q = 0) &= \frac{\lambda_0^N}{N !} e^{-\lambda_0},
\\
\mathbb{P}_{OOK}(N | X_q = 1) &= \frac{(\lambda_0 + \lambda_1)^N}{N !} e^{-(\lambda_0 + \lambda_1)},
\\
\mathbb{P}_{2-PPM}(N_1, N_2 | X_q = 0, \tau) &= \frac{\tau^{N_1} \lambda_0^{N_1} (1-\tau)^{N_2} (\lambda_0 + \lambda_1)^{N_2}}{N_1 ! N_2 !} e^{-\tau\lambda_0 - (1-\tau)(\lambda_0 + \lambda_1)},
\\
\mathbb{P}_{2-PPM}(N_1, N_2 | X_q = 1, \tau) &= \frac{\tau^{N_1} (\lambda_0 + \lambda_1)^{N_1} (1-\tau)^{N_2} \lambda_0^{N_2}}{N_1 ! N_2 !} e^{-\tau(\lambda_0 + \lambda_1) - (1-\tau)\lambda_0};
\end{aligned}
\end{equation}
$\lambda_1$ denotes the mean number of detected photoelectrons in each symbol duration; $N, N_1, N_2$ denote the number of received photoelectrons; and $\tau$ denotes the duty ratio of the pulse in each symbol duration for 2-PPM.
The achievable rates of OOK and 2-PPM modulation are compared in Figure \ref{fig:OOK_vs_PPM}, where OOK modulation shows higher achievable rate.
\begin{figure}
\centering
	\includegraphics[width=0.7\textwidth]{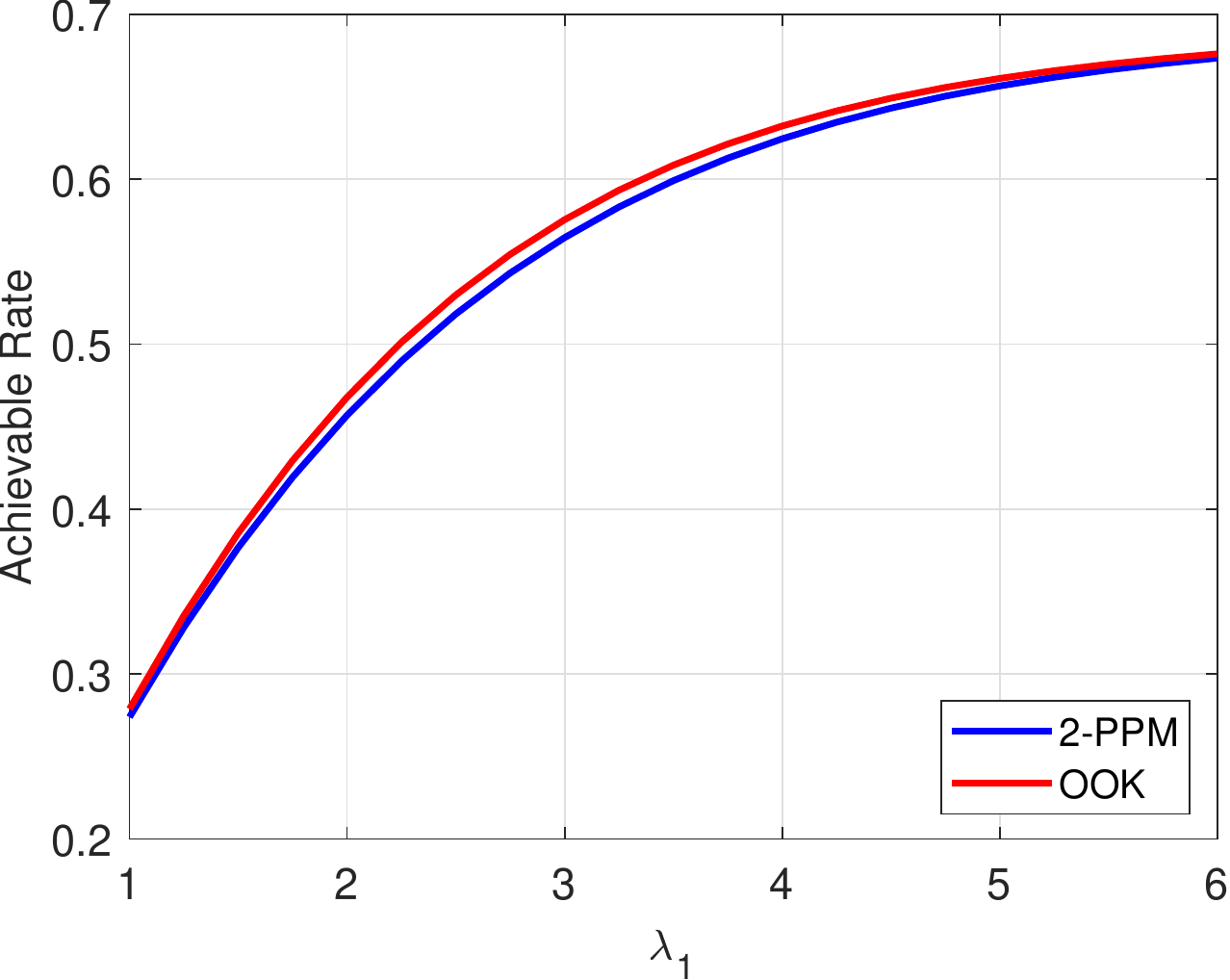}
    \caption{\label{fig:OOK_vs_PPM} The comparison between OOK and 2-PPM modulations with background intensity $1 \times 10^4$ per second.}
\end{figure}
}

\textcolor{red}{
\subsection{Proof of State Transition Matrix}
}

\textcolor{red}{
For $k = (t \text{ mod }L) + 1$, we have $\lceil \frac{t+1-k+1}{L} \rceil = \lceil \frac{t-k+1}{L} \rceil + 1$, and $\lceil \frac{t+1-r+1}{L} \rceil = \lceil \frac{t-r+1}{L} \rceil$ for $r \neq k$.
Due to $\boldsymbol{S}_t = [z_{1, \lceil \frac{t}{L} \rceil}, z_{2, \lceil \frac{t-1}{L} \rceil}, \dots, z_{L, \lceil \frac{t-L+1}{L} \rceil}]^T$, the $r$-th element of $\boldsymbol{S}_t$ and $\boldsymbol{S}_{t+1}$ must satisfy
$z_{r, \lceil \frac{t-r+1}{L} \rceil} = z_{r, \lceil \frac{t+1-r+1}{L} \rceil}$ for $r \neq k$.
Consequently, The state transition probability $\mathbb{P} (\boldsymbol{S}_{t+1} = \boldsymbol{s}_{t+1,j} | \boldsymbol{S}_t = \boldsymbol{s}_{t,i}) = 0$, if $z_{r, \lceil \frac{t-r+1}{L} \rceil} \neq z_{r, \lceil \frac{t+1-r+1}{L} \rceil}$; and $\mathbb{P} (\boldsymbol{S}_{t+1} = \boldsymbol{s}_{t+1,j} | \boldsymbol{S}_t = \boldsymbol{s}_{t,i}) = \mathbb{P} (z_{k, \lceil \frac{t+1-k+1}{L} \rceil} | z_{k, \lceil \frac{t-k+1}{L} \rceil})$.
Furthermore, $z_{k, \lceil \frac{t+1-k+1}{L} \rceil}$ is independent with $z_{k, \lceil \frac{t-k+1}{L} \rceil}$, hence we have $\mathbb{P} (z_{k, \lceil \frac{t+1-k+1}{L} \rceil} | z_{k, \lceil \frac{t-k+1}{L} \rceil}) = q_{k, \lceil \frac{t+1-k+1}{L} \rceil}^{ z_{k, \lceil \frac{t+1-k+1}{L} \rceil}} (1-q_{k, \lceil \frac{t+1-k+1}{L} \rceil})^{z_{k, \lceil \frac{t+1-k+1}{L} \rceil}}$.
In addition, $z_{i, \lceil \frac{t-i+1}{L} \rceil} = \boldsymbol{S}_{t} \cdot \boldsymbol{e}_i$ for $1 \leq i \leq L$.
Therefore, $z_{r, \lceil \frac{t-r+1}{L} \rceil} \neq z_{r, \lceil \frac{t+1-r+1}{L} \rceil}$ is equivalent with $\big( \boldsymbol{s}_{t+1,j} \cdot \boldsymbol{e}_r \big) \odot \big( \boldsymbol{s}_{t,i} \cdot \boldsymbol{e}_r \big) = 0$, and $\mathbb{P} (\boldsymbol{S}_{t+1} = \boldsymbol{s}_{t+1,j} | \boldsymbol{S}_t = \boldsymbol{s}_{t,i})$ can be simplified into Equation (\ref{eq:atij}).}

\subsection{Proof of Chain Rules on Conditional Probabilities}

We prove the proposition based on the following chain rule on the probability of received signal given two users since the numbers of received photoelectrons in different chips are independent of each other,
\begin{equation}
\begin{aligned}
\mathbb{P} (\boldsymbol{N}_T | \boldsymbol{Z}_{\mathcal{L}}) &= \prod^{T}_{t=1} \mathbb{P} (N_{t} | Z_{1,\lceil \frac{t}{L}\rceil}, Z_{2,\lceil \frac{t-1}{L}\rceil}, \dots, Z_{L,\lceil \frac{t-L+1}{L}\rceil} ).
\end{aligned}
\end{equation}
Consequently, Equation (\ref{eq:Prop1}) can be proved by
\begin{equation}
\begin{aligned}
\label{eq:conditional_aposterior_probability}
\mathbb{P} (\boldsymbol{Z}_k | \boldsymbol{Z}_{\mathcal{L}\backslash k}, \boldsymbol{N}_T)
&= \mathbb{P} \big(Z_{k, 1}, \boldsymbol{Z}_{k, [2, L]} | \{ Z_{i,\lceil \frac{{t_1}-i+1}{L} \rceil} \}, \{ N_{{t_1}} \}, \{ Z_{i,\lceil \frac{{\tilde{t}}-i+1}{L} \rceil} \}, \{ N_{{\tilde{t}}}  \} \big)
\\
&= \frac{\mathbb{P} \big(Z_{k, 1}, \boldsymbol{Z}_{k, [2, L]}, \{ N_{{t_1}} \}, \{ N_{{\tilde{t}}} \} \mid \{ Z_{i,\lceil \frac{{t_1}-i+1}{L} \rceil} \}, \{ Z_{i,\lceil \frac{{\tilde{t}}-i+1}{L} \rceil}   \} \big)}{\sum_{Z_{k, 1}} \sum_{\boldsymbol{Z}_{k, [2, L]}} \mathbb{P} \big(Z_{k, 1}, \boldsymbol{Z}_{k, [2, L]}, \{ N_{{t_1}} \}, \{ N_{{\tilde{t}}} \} \mid \{ Z_{i,\lceil \frac{{t_1}-i+1}{L} \rceil} \}, \{ Z_{i,\lceil \frac{{\tilde{t}}-i+1}{L} \rceil}   \} \big)}
\\
&= \frac{\mathbb{P} \big( \{ N_{{t_1}} \} \mid Z_{k, 1}, \{ Z_{i,\lceil \frac{{t_1}-i+1}{L} \rceil} \} \big) \mathbb{P} \big(Z_{k, 1} \big) \mathbb{P} \big( \{ N_{{\tilde{t}}} \} \mid \boldsymbol{Z}_{k, [2, L]}, \{ Z_{i,\lceil \frac{{\tilde{t}}-i+1}{L} \rceil} \} \big) \mathbb{P} \big( \boldsymbol{Z}_{k, [2, L]} \big)}
{\sum_{Z_{k, 1}} \sum_{\boldsymbol{Z}_{k, [2, L]}} \mathbb{P} \big( \{ N_{{t_1}} \} \mid Z_{k, 1}, \{ Z_{i,\lceil \frac{{t_1}-i+1}{L} \rceil} \} \big) \mathbb{P} \big(Z_{k, 1} \big) \mathbb{P} \big( \{ N_{{\tilde{t}}} \} \mid \boldsymbol{Z}_{k, [2, L]}, \{ Z_{i,\lceil \frac{{\tilde{t}}-i+1}{L} \rceil} \} \big) \mathbb{P} \big( \boldsymbol{Z}_{k, [2, L]} \big) }
\\
&= \frac{\mathbb{P} \big( \{ N_{{t_1}} \} \mid Z_{k, 1}, \{ Z_{i,\lceil \frac{{t_1}-i+1}{L} \rceil} \} \big) \mathbb{P} \big(Z_{k, 1} \big)}{\sum_{Z_{k, 1}} \mathbb{P} \big( \{ N_{{t_1}} \} \mid Z_{k, 1}, \{ Z_{i,\lceil \frac{{t_1}-i+1}{L} \rceil} \} \big) \mathbb{P} \big(Z_{k, 1} \big)}
\frac{\mathbb{P} \big( \{ N_{{\tilde{t}}} \} \mid \boldsymbol{Z}_{k, [2, L]}, \{ Z_{i,\lceil \frac{{\tilde{t}}-i+1}{L} \rceil} \} \big) \mathbb{P} \big( \boldsymbol{Z}_{k, [2, L]} \big)}{
\sum_{\boldsymbol{Z}_{k, [2, L]}}  \mathbb{P} \big( \{ N_{{\tilde{t}}} \} \mid \boldsymbol{Z}_{k, [2, L]}, \{ Z_{i,\lceil \frac{{\tilde{t}}-i+1}{L} \rceil} \} \big) \mathbb{P} \big( \boldsymbol{Z}_{k, [2, L]} \big) }
\\
&= \mathbb{P} \big( Z_{k, 1} \mid \{ N_{{t_1}} \}, \{ Z_{i,\lceil \frac{{t_1}-i+1}{L} \rceil} \} \big)
\mathbb{P} \big( \boldsymbol{Z}_{k, [2, L]} \mid \{ N_{{\tilde{t}}} \}, \{ Z_{i,\lceil \frac{{\tilde{t}}-i+1}{L} \rceil} \} \big),
\end{aligned}
\end{equation}
where $\boldsymbol{Z}_{k, [2, M]} = [Z_{k, 2}, Z_{k, 2}, \dots, Z_{k, M}]$, and the indexes involved in the brackets $i \in \mathcal{L} \backslash k$, $k \leq t_1 \leq k+L-1$ and $k+L \leq \tilde{t} \leq k+ML-1$.
Re-factorizing Equation (\ref{eq:conditional_aposterior_probability}), we have
\begin{equation}
\begin{aligned}
\label{eq:Prop3}
\mathbb{P} (\boldsymbol{Z}_k | \boldsymbol{Z}_{\mathcal{M}\backslash k}, \boldsymbol{N}_T)
&= \mathbb{P} \big( Z_{k, 1} \mid \{ N_{{t_1}} \}, \{ Z_{i,\lceil \frac{{t_1}-i+1}{L} \rceil} \} \big)
\mathbb{P} \big( Z_{k, 2} \mid \{ N_{{t_2}} \}, \{ Z_{i,\lceil \frac{{t_2}-i+1}{L} \rceil} \} \big)
\mathbb{P} \big( \boldsymbol{Z}_{k, [3, M]} \mid \{ N_{\tilde{t}} \}, \{ Z_{i,\lceil \frac{\tilde{t}-i+1}{L} \rceil} \} \big),
\end{aligned}
\end{equation}
where the indexes involved in the brackets $i \in \mathcal{L} \backslash k$, $k+L \leq t_2 \leq k+2L-1$ and $k+2L \leq \tilde{t} \leq k+ML-1$.
Re-factorize Equation (\ref{eq:Prop3}) for $M-2$ times, we have
\begin{equation}
\begin{aligned}
\mathbb{P} (\boldsymbol{Z}_k | \boldsymbol{Z}_{\mathcal{L}\backslash k}, \boldsymbol{N}_T) = \prod^{M}_{j=1} \mathbb{P} \big(Z_{k, j} | \{ Z_{i,\lceil \frac{t_j-i+1}{L} \rceil} \}, \{ N_{t_j} \} \big),
\end{aligned}
\end{equation}
where the indexes involved in the brackets $i \in \mathcal{L} \backslash k$ and $k+(j-1)L \leq t_j \leq k+jL-1$.

\subsection{Proof of conditional entropies}

We prove the proposition by Equation (\ref{eq:entropy_2users_proof}) based on Proposition 1.
\begin{equation}
\begin{aligned}
\label{eq:entropy_2users_proof}
\mathrm{H} (\boldsymbol{Z}_k | \boldsymbol{Z}_{\mathcal{L}\backslash k}, \boldsymbol{N}_T)
&= \sum_{\!\!\!\!\!\!\!\!\!\! \boldsymbol{Z}_{\mathcal{L}} \in \mathscr{B}^{ML}} \sum_{\boldsymbol{N}_T \in \mathbb{N}^T} \!\!\! \mathbb{P} (\boldsymbol{Z}_{\mathcal{L}}, \boldsymbol{N}_T) \log_2  \mathbb{P} (\boldsymbol{Z}_k | \boldsymbol{Z}_{\mathcal{L}\backslash k}, \boldsymbol{N}_T)
\\
&= \sum_{\boldsymbol{Z}_{\mathcal{L}} \in \mathscr{B}^{ML}} \mathbb{P} (\boldsymbol{Z}_{\mathcal{L}}) \!\! \sum_{\boldsymbol{N}_T \in \mathbb{N}^T} \!\!\! \mathbb{P} (\boldsymbol{N}_T | \boldsymbol{Z}_{\mathcal{L}}) \log_2  \mathbb{P} (\boldsymbol{Z}_k | \boldsymbol{Z}_{\mathcal{L}\backslash k}, \boldsymbol{N}_T)
\\
&= \sum_{\boldsymbol{Z}_{\mathcal{L}} \in \mathscr{B}^{ML}} \mathbb{P} (\boldsymbol{Z}_{\mathcal{L}}) \!\! \sum_{\boldsymbol{N}_T \in \mathbb{N}^T} \!\!\! \mathbb{P} (\boldsymbol{N}_T | \boldsymbol{Z}_{\mathcal{L}}) \log_2  \prod^{M}_{j=1} \mathbb{P} \big(Z_{k, j} | \{ Z_{i,\lceil \frac{t_j-i+1}{L} \rceil}\}, \{N_{t_j} \} \big)
\\
&= \sum_{\boldsymbol{Z}_{\mathcal{L}} \in \mathscr{B}^{ML}} \mathbb{P} (\boldsymbol{Z}_{\mathcal{L}}) \!\! \sum_{\boldsymbol{N}_T \in \mathbb{N}^T} \!\!\! \mathbb{P} (\boldsymbol{N}_T | \boldsymbol{Z}_{\mathcal{L}}) \log_2 \prod^{M}_{j=1} \frac{\mathbb{P} \big(Z_{k, j}, \{N_{t_j}\} | \{ Z_{i,\lceil \frac{t_j-i+1}{L} \rceil}\} \big)}{\mathbb{P} \big( \{N_{t_j}\} | \{ Z_{i,\lceil \frac{t_j-i+1}{L} \rceil}\} \big)}
\\
&= \sum^{M}_{j=1} \sum_{ Z_{k,j} \in \mathscr{B}} \Big( \!\!\! \sum_{Z_{i, \lceil (t-i+1)/L \rceil}\in\mathscr{B}} \!\!\!\!\!\! \Big)^{k+(j-1)L \leq t \leq k+jL-1 \hfill \atop i\in\mathcal{L}\backslash k \hfill} \!\!\!\!\!\!\!\!\!\!\!\!\!\!\!\! \mathbb{P}(Z_{k,j}) \Bigg[ \prod_{ i\in\mathcal{L}\backslash k}  \mathbb{P} (\{ Z_{i,\lceil \frac{t_j-i+1}{L} \rceil} \}) \Bigg] \sum_{\{ N_{t_j} \} \in \mathbb{N}^{L}} \!\!\! \mathbb{P} (\{ N_{t_j} \} | Z_{k,j}, \{ Z_{i,\lceil \frac{t_j-i+1}{L} \rceil} \})
\\
&\log_2  \frac{\mathbb{P} \big(Z_{k, j}, \{N_{t_j}\} | \{ Z_{i,\lceil \frac{t_j-i+1}{L} \rceil}\} \big)}{\mathbb{P} \big( \{N_{t_j}\} | \{ Z_{i,\lceil \frac{t_j-i+1}{L} \rceil}\} \big)}
\\
&= \sum^{M}_{j=1} \sum_{ Z_{k,j} \in \mathscr{B}} \Big( \!\!\! \sum_{Z_{i, \lceil (t-i+1)/L \rceil}\in\mathscr{B}} \!\!\!\!\!\! \Big)^{k+(j-1)L \leq t \leq k+jL-1 \hfill \atop i\in\mathcal{L}\backslash k \hfill} \!\!\!\!\!\!\!\!\!\!\!\!\!\!\!\! \mathbb{P}(Z_{k,j}) \Bigg[ \prod_{ i\in\mathcal{L}\backslash k} \prod_{t=k+(j-1)L}^{k+jL-1} \mathbb{P} (Z_{i,\lceil \frac{t-i+1}{L} \rceil}) \Bigg] \sum_{ \{ N_{t_j} \} \in \mathbb{N}^L} \Bigg[ \prod_{t=k+(j-1)L}^{k+jL-1} \mathbb{P} (N_t | Z_{k,j}, \{Z_{i,\lceil \frac{t-i+1}{L} \rceil}\}) \Bigg]
\\
&\log_2  \frac{\mathbb{P} \big(Z_{k, j}, \{N_{t_j}\} | \{ Z_{i,\lceil \frac{t_j-i+1}{L} \rceil}\} \big)}{\mathbb{P} \big( \{N_{t_j}\} | \{ Z_{i,\lceil \frac{t_j-i+1}{L} \rceil}\} \big)}
\\
&= \sum^{M}_{j=1} \sum_{ Z_{k,j} \in \mathscr{B}} \Big( \!\!\! \sum_{Z_{i, \lceil (t-i+1)/L \rceil}\in\mathscr{B}} \!\!\!\!\!\! \Big)^{k+(j-1)L \leq t \leq k+jL-1 \hfill \atop i\in\mathcal{L}\backslash k \hfill} \!\!\!\!\!\!\!\!\!\!\!\!\!\!\!\! \mathbb{P}(Z_{k,j}) \Bigg[ \prod_{ i\in\mathcal{L}\backslash k} \prod_{t=k+(j-1)L}^{k+jL-1} \mathbb{P} (Z_{i,\lceil \frac{t-i+1}{L} \rceil}) \Bigg] \sum_{ \{ N_{t_j} \} \in \mathbb{N}^L} \Bigg[ \prod_{t=k+(j-1)L}^{k+jL-1} \mathbb{P} (N_t | Z_{k,j}, \{Z_{i,\lceil \frac{t-i+1}{L} \rceil}\}) \Bigg]
\\
&\log_2  \frac{\mathbb{P} \big(Z_{k,j}\big) \prod_{t=k+(j-1)L}^{k+jL-1} \mathbb{P} (N_t | Z_{k,j}, \{Z_{i,\lceil \frac{t-i+1}{L} \rceil}\})}{\sum_{ Z_{k,j} \in \mathscr{B}} \mathbb{P} \big(Z_{k,j}\big) \prod_{t=k+(j-1)L}^{k+jL-1} \mathbb{P} (N_t | Z_{k,j}, \{Z_{i,\lceil \frac{t-i+1}{L} \rceil}\})}
\end{aligned}
\end{equation}

Typically, for single user transmission, the prior probability of the transmitted symbols remains constant, i. e, $q_{i,j} = q$ for $1 \leq i \leq L$ and $1 \leq j \leq M$.
Consequently, each term in $\sum^{M}_{j=1} [\bullet]$ remain constant for $2 \leq j \leq M-1$.
When $j=1$ or $j=M$, $\lceil \frac{t-i+1}{L} \rceil$ may equal $0$ or $M+1$, we define $Z_{i, 0} = Z_{i, M+1} = 0$ for $1 \leq i \leq L$ due to the finite number of transmitted symbols.
Neglecting the effect of $j=1$ and $j=M$, we have that $\sum^{M}_{j=1} \sum_{ Z_{k,j} \in \mathscr{B}} [\bullet] = M \sum_{ Z_{k,j_0} \in \mathscr{B}} [\bullet]$, where $j_0$ can take any integer value in $[2, M-1]$; and
\begin{equation}
\Big\lceil \frac{t-i+1}{L} \Big\rceil =\left\{
\begin{array}{ll}
j_0, & \textrm{if $1 \leq i < k, k+(j_0-1)L \leq t \leq i + j_0 L - 1$}; \\
j_0+1, & \textrm{if $1 \leq i < k, i + j_0L \leq t \leq k + j_0L - 1$}; \\
j_0-1, & \textrm{if $k < i \leq L, k+(j_0-1)L \leq t \leq i + (j_0-1) L - 1$}; \\
j_0, & \textrm{if $k < i \leq L, i + (j_0-1)L \leq t \leq k + j_0L - 1$}.
\end{array}
\right.
\end{equation}

Letting $j_0 = 2$, we have the following simplified form of Equation (\ref{eq:entropy_2users_proof}),
\begin{equation}
\begin{aligned}
\label{eq:entropy_2users_proof2}
\mathrm{H} (\boldsymbol{Z}_k | \boldsymbol{Z}_{\mathcal{L}\backslash k}, \boldsymbol{N}_T)
&= M \sum_{ Z_{k} \in \mathscr{B}} \bigg( \sum_{ Z_{i,2} \in \mathscr{B}} \sum_{ Z_{i,3} \in \mathscr{B}} \bigg)^{1 \leq i < k}
\bigg( \sum_{ Z_{i,1} \in \mathscr{B}} \sum_{ Z_{i,2} \in \mathscr{B}} \bigg)^{k < i \leq L} \mathbb{P}(Z_{k}) \Bigg[ \prod_{1 \leq i < k} \mathbb{P} (Z_{i,2}) \mathbb{P} (Z_{i,3}) \prod_{k < i \leq L} \mathbb{P} (Z_{i,1}) \mathbb{P} (Z_{i,2}) \Bigg]
\\
&\!\!\!\!\!\!\!\!\sum_{ \{ N_{k+L}, \ldots, N_{k+2L-1} \} \in \mathbb{N}^L} \Bigg[ \prod_{t=k+L}^{k+2L-1} \mathbb{P} (N_t | Z_k, \{Z_{i,\lceil \frac{t-i+1}{L} \rceil}\}) \Bigg]
\log_2  \frac{\mathbb{P} \big(Z_{k}\big) \prod_{t=k+L}^{k+2L-1} \mathbb{P} (N_t | Z_{k}, \{Z_{i,\lceil \frac{t-i+1}{L} \rceil}\})}{\sum_{ Z_{k} \in \mathscr{B}} \mathbb{P} \big(Z_{k}\big) \prod_{t=k+L}^{k+2L-1} \mathbb{P} (N_t | Z_{k}, \{Z_{i,\lceil \frac{t-i+1}{L} \rceil}\})}.
\end{aligned}
\end{equation}

\textcolor{red}{
\subsection{Derivation of $\hat{\boldsymbol{\Lambda}}^{(v)}$  in the M-Step of Channel Estimation}}

The likelihood function is given by
\begin{equation}
\begin{aligned}
\label{eq:EM_M_proof1}
\mathcal{L} (\boldsymbol{N}^p | \lambda_{\boldsymbol{s}_i} = \hat{\lambda}^{(v)}_{\boldsymbol{s}_i})
&= \sum^{T_p}_{t=1} \log \sum_{\boldsymbol{s}_i \in \mathscr{B}^{L\setminus L_p}} \mathbb{P} (\boldsymbol{N}^p, \boldsymbol{S}^p_t = \boldsymbol{s}_i | \lambda_{\boldsymbol{s}_i} = \hat{\lambda}^{(v)}_{\boldsymbol{s}_i})
\\
&\geq \sum^{T_p}_{t=1} \sum_{\boldsymbol{s}_i \in \mathscr{B}^{L\setminus L_p}} Q^{(v)}(\boldsymbol{S}^p_t = \boldsymbol{s}_i) \log \frac{\mathbb{P} (\boldsymbol{N}^p, \boldsymbol{S}^p_t = \boldsymbol{s}_i | \lambda_{\boldsymbol{s}_i} = \hat{\lambda}^{(v)}_{\boldsymbol{s}_i})}{Q^{(v)}(\boldsymbol{S}^p_t = \boldsymbol{s}_i)}.
\end{aligned}
\end{equation}
Letting $\tilde{\mathcal{L}}^{(v)} (\boldsymbol{N}^p | \lambda_{\boldsymbol{s}_i} = \hat{\lambda}^{(v)}_{\boldsymbol{s}_i})$ denote the last term of above inequality, we have that
\begin{equation}
\begin{aligned}
&\tilde{\mathcal{L}}^{(v)} (\boldsymbol{N}^p | \lambda_{\boldsymbol{s}_i} = \hat{\lambda}^{(v)}_{\boldsymbol{s}_i})
\sim \!\! \sum^{T_p}_{t=1} \!\! \sum_{\boldsymbol{s}_i \in \mathscr{B}^{L\setminus L_p}} \!\! Q^{(v)}(\boldsymbol{S}^p_t = \boldsymbol{s}_i) \Big( N^p_{t} \log \tau_t \hat{\lambda}^{(v)}_{\boldsymbol{s}_i} \!-\! \log N^p_{t} ! - \tau_t \hat{\lambda}^{(v)}_{\boldsymbol{s}_i} \Big).
\nonumber
\end{aligned}
\end{equation}
Hence, the partial derivative of likelihood function is given by
\begin{equation}
\begin{aligned}
&\frac{\partial}{\partial \lambda^{(v)}_{\boldsymbol{s}_i}} \tilde{\mathcal{L}}^{(v)} (\boldsymbol{N}^p = \boldsymbol{n} | \lambda_{\boldsymbol{s}_i} = \hat{\lambda}^{(v)}_{\boldsymbol{s}_i})
= \sum^{T_p}_{t=1} Q^{(v)}(\boldsymbol{S}^p_t = \boldsymbol{s}_i) \Bigg( \frac{N^p_{t}}{\hat{\lambda}^{(v)}_{\boldsymbol{s}_i}} - \tau_t \Bigg).
\end{aligned}
\end{equation}
Letting $\frac{\partial}{\partial \lambda_{\boldsymbol{s}_j}} \mathcal{L} (\boldsymbol{N}^p = \boldsymbol{n} | \lambda_{\boldsymbol{s}_j} = \hat{\lambda}^{(v)}_{\boldsymbol{s}_j}) = 0$, we have that
\begin{equation}
\begin{aligned}
\label{eq:EM_M_proof2}
\hat{\lambda}^{(v)}_{\boldsymbol{s}_j} &= \arg \max \tilde{\mathcal{L}}^{(v)} (\boldsymbol{N}^p = \boldsymbol{n} | \lambda_{\boldsymbol{s}_i} = \hat{\lambda}^{(v)}_{\boldsymbol{s}_i})
= \frac{ \sum^{T_p}_{t=1} Q^{(v)}(\boldsymbol{S}^p_t = \boldsymbol{s}_i) N^p_{t} }{ \sum^{T_p}_{t=1} Q^{(v)}(\boldsymbol{S}^p_t = \boldsymbol{s}_i) \tau_t}.
\end{aligned}
\end{equation}

\bibliographystyle{ieeetr}
\bibliography{random_access_reference_abbr}

\end{spacing}
\end{document}